\documentclass[journal]{IEEEtran}

\usepackage{amsmath,amssymb,amsthm}
\usepackage{xifthen}
\usepackage{mathtools}
\usepackage{enumerate}
\usepackage{microtype}
\usepackage{xspace}
\usepackage{bm}
\usepackage[T1]{fontenc}
\usepackage{fancyhdr}
\usepackage{lastpage}
\usepackage{bbm}

\newcommand{\mbs}[1]{\pmb{#1}}
\newcommand{\vect}[1]{{\lowercase{\mbs{#1}}}}
\newcommand{\mat}[1]{{\uppercase{\mbs{#1}}}}

\newcommand{\T}{{\scriptscriptstyle\mathsf{T}}}
\renewcommand{\H}{{\scriptscriptstyle\mathsf{H}}}

\renewcommand{\Re}[1][]{\ifthenelse{\isempty{#1}}{\operatorname{Re}}{\operatorname{Re}\left(#1\right)}}
\renewcommand{\Im}[1][]{\ifthenelse{\isempty{#1}}{\operatorname{Im}}{\operatorname{Im}\left(#1\right)}}

\newcommand{\SNR}{\mathsf{snr}}

\newcommand{\rv}{\vect{r}}

\newcommand{\vv}{\vect{v}}

\newcommand{\Mm}{\mat{M}}

\newcommand{\Cc}{{\mathcal C}}

\newcommand{\Nc}{{\mathcal N}}

\newcommand{\Sc}{{\mathcal S}}

\newcommand{\CC}{\mathbb{C}}

\newcommand{\Id}{\mat{\mathrm{I}}}

\newcommand{\CN}[1][]{\ifthenelse{\isempty{#1}}{\mathcal{N}_{\mathbb{C}}}{\mathcal{N}_{\mathbb{C}}\left(#1\right)}}

\renewcommand{\P}[1][]{\ifthenelse{\isempty{#1}}{\mathbb{P}}{\mathbb{P}\left(#1\right)}}
\newcommand{\E}[1][]{\ifthenelse{\isempty{#1}}{\mathbb{E}}{\mathbb{E}\left[#1\right]}}
\newcommand{\Var}[1][]{\ifthenelse{\isempty{#1}}{\mathbb{E}}{\mathsf{Var}\left[#1\right]}}
\newcommand{\I}[1][]{\ifthenelse{\isempty{#1}}{\mathbb{I}}{\mathbb{I}\left\{#1\right\}}}
\renewcommand{\det}[1][]{\ifthenelse{\isempty{#1}}{\mathrm{det}}{\text{det}\left(#1\right)}}
\newcommand{\trace}[1][]{\ifthenelse{\isempty{#1}}{\mathrm{tr}}{\text{tr}\left(#1\right)}}
\newcommand{\rank}[1][]{\ifthenelse{\isempty{#1}}{\mathrm{rank}}{\text{rank}\left(#1\right)}}
\newcommand{\diag}[1][]{\ifthenelse{\isempty{#1}}{\mathrm{diag}}{\text{diag}\left(#1\right)}}
\newcommand{\Cov}[1][]{\ifthenelse{\isempty{#1}}{\mathsf{Cov}}{\mathsf{Cov}\left(#1\right)}}

\newcommand{\defeq}{\triangleq}
\newcommand{\eqdef}{\triangleq}

\newtheorem{proposition}{Proposition}
\newtheorem{remark}{Remark}[section]
\newtheorem{definition}{Definition}
\newtheorem{theorem}{Theorem}
\newtheorem{example}{Example}
\newtheorem{lemma}{Lemma}

\newcounter{enumi_saved}
\usepackage{answers}
\Newassociation{solution}{Solution}{solutionfile}

\AtBeginDocument{\Opensolutionfile{solutionfile}[\jobname]}
\AtEndDocument{\Closesolutionfile{solutionfile}\clearpage
}

\IfFileExists{MinionPro.sty}{
}{
}

\usepackage{subfigure}
\usepackage{multirow}
\usepackage{array}
\usepackage{amsmath}
\usepackage[flushleft]{threeparttable}
\usepackage{pgfplots}
\usepackage{epstopdf}
\usepackage{mathtools}
\mathtoolsset{showonlyrefs}

\renewcommand{\rv}[1]{{\mathrm{#1}}}
\newcommand{\rvVec}[1]{\pmb{\mathrm{#1}}}
\newcommand{\rvMat}[1]{\pmb{\mathsf{#1}}}
\newcommand{\const}{\eta}

\newcommand{\nt}{{n_{\text{t}}}}

\newcommand{\normmem}{m}
\newcommand{\Ruc}{\mathcal{R}_{\rm {uni}}}
\newcommand{\Rmc}{\mathcal{R}_{\rm {mul}}}

\renewcommand{\defeq}{:=}
\renewcommand{\eqdef}{=:}

\newcommand{\SINR}{\rv{SINR}}
\renewcommand{\SNR}{\rv{SNR}}
\newcommand{\SNRmean}{\overline{\SNR}}

\newcommand{\normh}{\frac{\|\rvVec{H}_k\|^2}{n_t}}
\newcommand{\minSNRo}{\min_{k\in[K]} \SNR_k }

\newcommand{\ntgelnK}{\nt\ge\ln(K) + O(1)}
\newcommand{\LntgelnK}{L\nt \ge \ln(K) + O(1)}

\newcommand{\approxasympt}{\sim}

\usepackage{booktabs}
\usepackage{array}

\title{Scalable Content Delivery with Coded Caching in Multi-Antenna Fading Channels} 
\author{
Khac-Hoang Ngo,~\IEEEmembership{Student Member,~IEEE,} Sheng Yang,~\IEEEmembership{Member,~IEEE,}~and\\
Mari Kobayashi,~\IEEEmembership{Senior Member,~IEEE}%
\thanks{The authors are with LSS, CentraleSup\'elec,  3 rue Joliot-Curie, 91190
Gif-sur-Yvette, France.~(e-mail: \texttt{\{khachoang.ngo, sheng.yang, mari.kobayashi\}@centralesupelec.fr})}%
\thanks{This paper was presented in part at the 2016 International Symposium on Turbo Codes \& Iterative Information Processing (Brest, France)~\cite{Yang2016contentdelivery}, the 2016 International Conference on Advanced Technologies for Communications (Hanoi, Vietnam)~\cite{hoang2016complementary}, and the 2016 54th Annual Allerton Conference on Communication, Control, and Computing (IL, USA)~\cite{hoang2015cacheaided}.}%
}

\begin{document}

\maketitle

\begin{abstract}
  We consider the content delivery problem in a fading multi-input single-output channel with cache-aided users. We are interested in the scalability of the equivalent content delivery rate when the number of users, $K$, is large. Analytical results show that, using coded caching and wireless multicasting, without channel state information at the transmitter~(CSIT), linear scaling of the content delivery rate with respect to $K$ can be achieved in some different ways. First, if the multicast transmission spans over $L$ independent sub-channels, e.g., in quasi-static fading if $L = 1$, and in block fading or multi-carrier systems if $L>1$, linear scaling can be obtained when the product of the number of transmit antennas and the number of sub-channels scales logarithmically with $K$. Second, even with a fixed number of antennas, we can achieve the linear scaling with a threshold-based user selection requiring only one-bit feedbacks from the users. When CSIT is available, we propose a mixed strategy that combines spatial multiplexing and multicasting. Numerical results show that, by optimizing the power split between spatial multiplexing and multicasting, we can achieve a significant gain of the content delivery rate with moderate cache size. 
\end{abstract}


\section{Introduction}
One critical issue in future wireless network is the expansion of wireless and mobile data traffic,
which is predicted to account for two-thirds of total data traffic by 2020~\cite{index2015global}. Massive MIMO, exploiting a huge number of antennas at the base station has been considered as a promising candidate to deal with the traffic expansion (see, e.g., \cite{larsson2014massive} and references therein). By creating parallel interference-free streams via spatial precoding (e.g. zero-forcing), multiple users can be simultaneously served. If the number of transmit antennas can scale with the number of users $K$, the total transmission time to serve $K$ users shall not increase with $K$ and the throughput of the system increases linearly with $K$. 
Another emerging solution, motivated by the ever-growing cheap on-board storage memory as well as the skewness of the video traffic, is {\em edge caching} \cite{maddah2013fundamental,maddahISIT2015Tutorial,bastug2014living,paschos2016wireless}. Namely, the traffic during peak hours can be substantially offloaded if we prefetch popular contents at the edge of the network. 
Recently, it has been shown by Maddah-Ali and Niesen that {\em coded caching} enables to achieve a constant number of 
total multicast transmissions to satisfy the demand of $K$ users
when $K$ is large \cite{maddah2013fundamental}. 
In contrast to parallel streams in massive MIMO, a careful design of cache placement enables to create a single stream which is simultaneously useful to multiple users. 

A common perception is that either massive MIMO or coded caching is {\em potentially} a scalable solution alone with respect to~(w.r.t.) the number of users. 
However, the scalability of these solutions actually relies on some ideal assumptions that may not hold in real systems. On one hand, the scalability of massive MIMO hinges on: 1)~the linearly increasing number of the transmit antennas w.r.t. the number of users, and 2)~the accuracy of CSIT. On the other hand, the scalability of coded caching relies on a \emph{non-vanishing} multicast rate of the underlying channel. It should be remarked that the pioneering work~\cite{maddah2013fundamental} and many follow-up extensions, e.g.,~\cite{maddah2013decentralized,niesen2013coded,pedarsani2013online,ji2015order},  ideally assumed an error-free shared link, which obviously fulfills the latter condition. Therefore, it is of practical and theoretical interest to address the following question from the engineering perspective: {\em is it beneficial to use both technologies?} 

In this paper, we investigate the scalability of the two solutions in the following simple setting. 
We consider the content delivery network where a $\nt$-antenna base station serves $K$ single-antenna users over an independent and identically distributed~(i.i.d.) Rayleigh fading downlink channel.
We consider $L$-parallel channel, where the transmission of a codeword spans over $L\ge 1$ interference-free sub-channels, such that in each sub-channel, the signal is perturbed by independent fading coefficients and independent noise. The case $L=1$, i.e., quasi-static fading channel, corresponds to low-mobility scenario or the latency constrained applications such as the video streaming with independently coded/decoded chunks. The case $L>1$ corresponds to higher mobility scenario or delay-tolerant applications where a codeword spans over a number of fading blocks, or multi-carrier systems where a codeword spans over a number of sub-carriers. 
Under this setting, we wish to study the complementary roles of massive MIMO (with spatial multiplexing) and coded caching (with multicasting). To this end, 
we define the \textit{equivalent content delivery rate} as a unified metric of the throughput performance. Our main focus is the {\em scalability}, i.e., the linear scaling of equivalent content delivery rate of two solutions in the large $K$ regime. The main findings 
  of the current work are three-fold and summarized below: 
\begin{enumerate}
	\item We reveal two different ways that can guarantee the scalability of the content delivery system without CSIT while building on multicasting and coded caching (Theorem~\ref{theo:scalableSchemes}):
	\begin{enumerate}
		\item using a large number $\nt$ of transmit antennas and/or spanning the transmission over a large number $L$ of independent sub-channels: we show that $\LntgelnK$ is sufficient;
		\item user selection scheduling with an arbitrary number of transmit antennas in quasi-static fading channel: we show that one-bit feedback is enough; 
	\end{enumerate}
	\item We show that massive MIMO with zero-forcing precoding using asymptotically more transmit antennas than users can also achieve linear delivery rate scaling as long as the CSIT error variance is bounded (Proposition~\ref{prop:Rsym}). 
	\item In order to further improve the overall content delivery rate, we propose to combine multicasting and spatial multiplexing with the optimal power split. The analysis together with numerical examples reveals that the proposed mixed scheme coincides with multicasting if the memory size is large enough or the total power is small (Proposition~\ref{prop:mixRate}, Remark~\ref{remark:optimalP0}). 
\end{enumerate}
We remark that some results (Propositions~\ref{prop:R0}, \ref{prop:R0blockfading}, \ref{prop:Rsym}, \ref{prop:mixRate}) in the paper can be used for rate scaling analysis of communication systems with multicasting and/or spatial multiplexing in general, i.e., even outside the scope of cache-aided content delivery.

The interplay between spatial multiplexing gain and coded caching gain in MIMO channels has been studied in recent works, including our earlier works~\cite{hoang2015cacheaided,hoang2016complementary,Yang2016contentdelivery}, as well as \cite{Shariatpanahi2017multiantenna,Elia2017,Piovano2017codedcachingMISO}. The work \cite{Shariatpanahi2017multiantenna}, following the ideas of \cite{Shariatpanahi2016multiserver}, proposed to deliver multiple coded multicast packets simultaneously by multiplexing them assuming full CSIT. The works \cite{Elia2017,Piovano2017codedcachingMISO} proposed to deliver one part of the requested file by multiplexing and the other part by multicasting in parallel using rate splitting. It is remarked that these works are different from ours in their underlying assumptions, designs, and objectives. First, when combining multicasting and spatial multiplexing, we focus mainly on the regime of massive MIMO where the number of transmit antennas grows with the number of users, while \cite{Shariatpanahi2017multiantenna,Elia2017,Piovano2017codedcachingMISO} study the case of $\nt \le K$. 
Second, our performance measure is the scaling of the long-term equivalent content delivery rate in the large $K$ regime. In \cite{Shariatpanahi2017multiantenna}, the similar content delivery rate is studied but focusing rather on the large SNR regime to see the degree-of-freedom gain, while the total transmission time is used in \cite{Elia2017,Piovano2017codedcachingMISO}. 
Finally, we restrict here to the off-the-shell placement strategies and assume that two independent information flows can be delivered by multicasting and multiplexing. On the other hand, \cite{Shariatpanahi2017multiantenna} and \cite{Elia2017,Piovano2017codedcachingMISO} propose some designs reflecting the network structure and CSIT quality, respectively, and let both multicasting and multiplexing contribute to one information flow. 


The remainder of the paper is organized as follows. The system model and performance metric is presented in
Section~\ref{sec:model}. Some mathematical preliminaries are provided in Section~\ref{sec:preliminaries}. The scalability
of the content delivery system with multicasting and with spatial multiplexing is discussed in
Section~\ref{sec:scalableMulticasting} and Section~\ref{sec:scalableMultiplexing}, respectively. In
Section~\ref{sec:mixed}, we propose the mixed delivery with simultaneous multicasting and spatial multiplexing to improve
the content delivery rate, and derive the optimal power split. Relevant numerical results are inserted in
Section~\ref{sec:scalableMulticasting} and Section~\ref{sec:mixed}. The paper is concluded in
Section~\ref{sec:conclusion}. Some of the proofs are presented in the main text whereas the more technical details are deferred to the appendix. 

{\em Notations:} For random variables, we use upper case non-italic letters, e.g., $\rv{X}$, for scalars, upper case non-italic bold letters, e.g., $\rvVec{V}$, for vectors, and upper case letter with bold and sans serif fonts, e.g.,
$\rvMat{M}$, for matrices.  Deterministic quantities are denoted 
with italic letters, e.g., a scalar $x$, a vector $\pmb{v}$, and a matrix $\pmb{M}$. 
The Euclidean norm of a vector is denoted by $\|\vv\|$. The transpose, conjugate, and conjugate
transpose of $\Mm$ are denoted $\Mm^\T$, $\Mm^*$, and $\Mm^\H$, respectively. We let $x^+ := \max\left\{x, 0 \right\}$. 
The indicator $\mathbbm{1}_{\{A\}}$ takes value $1$ if $A$ is true and $0$ otherwise. 
We use $[K]$ to denote the set of integers~$\{1,\ldots,K\}$. 
The convergence in distribution, in probability, and almost sure convergence are denoted $\xrightarrow{\text{d}}, \xrightarrow{\text{p}}$, and $\xrightarrow{\text{a.s.}}$, respectively.
${\rm Gamma}(k, \theta)$
denotes the Gamma distribution with shape $k$ and scale $\theta$, while ${\rm Exp}(\lambda)$ the exponential distribution with rate parameter
$\lambda$. The Gamma function is denoted by $\Gamma(x) = \int_{0}^{\infty}z^{x-1}e^{-z}dz$, while $\Gamma(x,t) = \int_{t}^{\infty}z^{x-1}e^{-z}dz$ and $\gamma(x,t)=\int_{0}^{t}z^{x-1}e^{-z}dz$ are the upper and lower incomplete Gamma
functions, respectively. 

The asymptotic notations $O, o, \Omega, \Theta, \approxasympt$ are w.r.t.~$K$, unless stated otherwise. Specifically, given two functions $f$ and $g$, we say: 1) $f(K) = O(g(K))$ if there exists a positive constant $c$ and an integer $K_0$ such that $|f(K)| \le c |g(K)|, \forall K \ge K_0$; 2) $f(K) = o(g(K))$ if $\lim\limits_{K\to\infty}\frac{f(K)}{g(K)} = 0$; 3) $f(K) = \Omega(g(K))$ if $g(K) = O(f(K))$; 4) $f(K) = \Theta(g(K))$ if both $f(K) = O(g(K))$ and $g(K) = O(f(K))$; and 5) $f(K)\approxasympt g(K)$ if ${\displaystyle\lim_{K\to\infty}} \frac{f(K)}{g(K)} = 1$.

\section{System model}
\label{sec:model}

\subsection{Content delivery model}
We consider a content delivery system 
where a content server is connected to $K$ users through a wireless downlink channel. This server has access to a library of $N$
files, assumed to be equally popular and with equal size $F$ bits for simplicity. 
Each user $k$ is equipped with a cache of size $MF$~bits, where $M \ge 1$ denotes the cache size measured in files. 
Prior to the actual request, each user can pre-fill their cache during off-peak hours, with supposedly negligible cost. 
We assume that each user $k\in [K]$ has a sequence of demands $d_k^{(1)},d_k^{(2)},\dots$, and one demand from each user is served at a time.
Upon the reception of a collection of $K$ requests from the users, and based on the cached contents
available to each user, the server encodes and sends the requested files through the delivery channel. 
In~\cite{maddah2013fundamental,maddah2013decentralized}, Maddah-Ali and Niesen proposed a caching/delivery scheme for error-free
multicast delivery channels. With such a scheme, known as \emph{coded caching}, the number of \emph{multicast} transmissions, normalized by the file size, needed to satisfy the demands of $K$ users is
\begin{multline} \label{eq:defT}
T(m,K) \defeq \\ \begin{cases}  
\left(1-m\right)\frac{1}{1/K+m},&\text{for centralized caching}, \\ 
\left(1-m\right)\frac{1-(1-m)^K}{m},&\text{for decentralized caching},
\end{cases} 
\end{multline}%
where $m \defeq \frac{M}{N}$ is the normalized cache memory.
The striking result is that the number of required transmissions converges to a constant as $K$ grows, i.e., $T(m,K) \xrightarrow{K\to\infty} \frac{1-m}{m}$, for both centralized and
decentralized caching. In other words, 
coded caching is
scalable in a system when the delivery channel is an error-free multicast channel.  

\subsection{Delivery channel model}
In this work, we consider a multi-antenna downlink channel, in which the content server is placed in a base station
with $\nt$ transmit antennas, and each of the $K$ users is equipped with a single antenna. We consider $L$-parallel channel in which the
transmission spans over $L$ independent coherence resource blocks (so-called sub-channels), with an emphasis on the case $L=1$.
When $L=1$, the channel is quasi-static fading
such that the channel coefficients remain unchanged during
the transmission of a whole coded block. In this case, receiver~$k$ at time $t$ has the observation 
\begin{align}
\rv{Y}_{k}[t] &= \rvVec{H}_k^\T  \,
\pmb{x}[t] + \rv{Z}_k[t], \quad t=1,2,\ldots,n, \label{eq:model}
\end{align}%
where $\pmb{x}[t]\in\mathbb{C}^{\nt\times1}$ is the input vector at time~$t$, with the average power
constraint $\frac{1}{n} \sum_{t=1}^n \|\pmb{x}[t]\|^2 \le P$; the additive noise process
$\{\rv{Z}_k[t]\}$ is assumed to be spatially and temporally white with normalized variance, i.e.,
$\rv{Z}_k[t] \sim \mathcal{CN}(0,1)$, $k\in [K]$. Since the additive noise power is normalized, the
transmit power $P$ is identified with the total signal-to-noise ratio (SNR) throughout the paper.
Hereafter, we omit the time index for simplicity. 
 For tractability, we assume that the channel is independent and symmetric across users with Rayleigh
fading, i.e., $\rvVec{H}_k \sim \mathcal{CN}(0, \Id_{\nt})$, $k\in [K]$. The whole channel matrix is denoted by $\rvMat{H} :=
[\rvVec{H}_1\ \cdots\ \rvVec{H}_K]^\T$. Then, in the general case $L\ge 1$, a codeword can span over $L$ sub-channels, such that both the fading coefficients and additive noise are independent but have the same statistics across sub-channels.

In practice, the channel state information (CSI) is not perfectly known at the transmitter, typically due to limited resource for uplink channel training in TDD~(time division duplex) or limited channel feedback bandwidth in FDD~(frequency division duplex). A common model for imperfect CSIT, modeling the minimum-mean-square-error (MMSE) channel estimation, is 
\begin{align}\label{eq:imperfectCSIT}
\rvMat{H} &= \hat{\rvMat{H}} + \tilde{\rvMat{H}},
\end{align}%
where $\hat{\rvMat{H}}$ and $\tilde{\rvMat{H}}$ are the mutually uncorrelated estimate and estimation error, with each entry of variance $1-\sigma^2$ and
$\sigma^2$, respectively. Since we assume Rayleigh fading, $\hat{\rvMat{H}}$ and $\tilde{\rvMat{H}}$ are independent and circularly symmetric Gaussian distributed. We assume that CSI is perfect at the receivers. 

\subsection{Equivalent content delivery rate and scalability}

In practice, we are interested in how fast the requested content can be available to the users. To that end, we formally define the long-term performance metric:
\begin{definition}
	The {\em equivalent content delivery rate}~(or, simply, content delivery rate or sum rate) is the number of total demanded information
	bits, including those already in the cache, that can be delivered per unit of time in average.
\end{definition}
For example, when $M=N$, the equivalent content delivery rate is $\infty$, since each user can have any content instantly.
Let $\bar{R}_0$ be the average multicast rate of the delivery channel in
 bits/second/Hz. To satisfy the demands of $K$ users, i.e., to \emph{complete} in total $K F$ demanded bits, we
need to send $T(\normmem,K) F$ bits, which takes $T(\normmem,K) F / \bar{R}_0$ units of time. It means that the equivalent content delivery rate of the system with coded caching is
\begin{align} \label{eq:Rmul}
  \Rmc = \frac{K}{T(\normmem,K)}\bar{R}_0(K,P) \quad \text{bits/second/Hz}. 
\end{align}%
Since the natural logarithm is more convenient for our purposes, we shall use ``nats'' instead of ``bits'' in the rest of
the paper, unless otherwise specified. Note that the formula~\eqref{eq:Rmul}, however, remains the same with a simple change of unit.

The system is {\em scalable} with the number $K$ of users if the equivalent content delivery rate scales at least linearly with $K$ when $K$ grows. With coded caching, it is enough to have a non-vanishing average multicast rate $\bar{R}_0(K,P)$.

\section{Some mathematical preliminaries}
\label{sec:preliminaries}

In this section, we provide some mathematical preliminaries that will be useful to
prove the main results. Sketches of proof will be provided in Appendix~\ref{app:preliminaries}.

\begin{lemma}[The Chernoff bound]
  For $N$ independent random variables $\rv{X}_n,~n = 1,\dots,N$, 
  \begin{align}
    \P[\sum_{n=1}^N \rv{X}_n \le x ] &\le e^{\nu x} \prod_{n=1}^N \E[e^{- \nu \rv{X}_n}], \label{eq:Chernoff} \\
    \P[\sum_{n=1}^N \rv{X}_n \ge x ] &\le e^{-\nu x} \prod_{n=1}^N \E[e^{\nu \rv{X}_n}], \quad \forall\,\nu>0.  \label{eq:Chernoff2}
  \end{align}%
\end{lemma}

\begin{lemma}\label{lemma:CDFgamma}
  Let $F_\rv{X}(x)$ be the cumulative distribution function~(CDF) of random variable $\rv{X}\sim{\rm Gamma}(\nt,\frac{1}{\nt})$,\footnote{We recall that if $\rv{X}\sim \text{Gamma}(n,a)$, then
  $\rv{X}$ is equivalent to the sum of $n$ i.i.d.~exponential random variables ${\rm Exp}(\frac{1}{a})$. }  then
  \begin{align}
    F_{\rv{X}}(\const) &\le e^{-\nt}, \quad\text{with }\const \approx 0.1586. \label{eq:CDFgamma}
  \end{align}%
\end{lemma}

\begin{lemma}\label{lemma:expectation}
  For $K$ i.i.d.~non-negative random variables $\rv{X}_k,~k = 1, \dots, K$, with the common CDF $F_{\rv{X}}(x)$,  
  \begin{align}
    \E[ \min_{k\in[K]}\rv{X}_k] \ge x_0 [1-F_{\rv{X}}(x_0)]^K,~\forall x_0
    \ge 0. \label{eq:Markov1}
  \end{align}%
  If  $F_{\rv{X}}(x)$ is strictly increasing, then for any $c>0$,
  \begin{align}
    \frac{\E[ \min_{k\in[K]}\rv{X}_k]}{F^{-1}_{\rv{X}}(\frac{c}{K})} \ge
    e^{-c} + o(1), \quad \text{when } K\to\infty. \label{eq:Markov2}
  \end{align}%
\end{lemma}

\begin{lemma} \label{lemma:log-asympt}
  For a sequence of nonnegative random variables $\{ \rv{X}_K \}$, when $K\to\infty$
  \begin{enumerate}
    \item if $\E[\rv{X}_K] = \Theta(1)$, then $\E[\ln(1+\rv{X}_K)]
      = \Theta(1)$;
    \item if $\E[\rv{X}_K] = o(1)$ and $\E[\rv{X}_K^2] = o(\E[\rv{X}_K])$, then $\E[\ln(1+\rv{X}_K)]
      \approxasympt \E[\rv{X}_K]$;
    \item if $\E[\rv{X}_K] = \Theta(g(K))$ with $g(K) \to \infty$ and $\Pr\left(\rv{X}_K \ge g(K)\right) \ge \rho$ for some constant $\rho > 0$, then $\E[\ln(1+\rv{X}_K)] = \Theta(\ln(1+g(K)))$.   
  \end{enumerate}
\end{lemma}
The set of random variables $\normh$, $k\in[K]$ are
i.i.d.~${\rm Gamma}(\nt,\frac{1}{\nt}$) with mean $1$. The following lemmas
describe the asymptotic behavior of the minimum value when $K$ is large.\footnote{Explicitly, $\E \left[\displaystyle\min_{k\in[K]} \normh\right] = \frac{1}{\nt}\sum_{i=0}^{K(\nt-1)} c_i i!K^{-i-1}$ where $c_i$ is defined recursively: $c_0 = 1$ and $c_i = \frac{1}{i} \sum_{j = 1}^{\min\{i,\nt-1\}} \frac{(K+1)j-i}{j!}c_{j-i}, \forall i\ge 1$. This close-form expression, however, does not bring further insights on the asymptotic behavior of $\E \left[\displaystyle\min_{k\in[K]} \normh\right]$.} 
\begin{lemma} \label{lemma:SNR0}
      When $\nt$ is fixed, as $K\to\infty$, the random variable $a_K{\displaystyle\min_{k\in[K]}} \normh$ with $a_K \defeq \nt\left(\frac{K}{\nt!}\right)^{1/\nt}$
  converges \emph{of mean}\footnote{The convergence {\em of mean} of a sequence of random variables $\{\rv{Y}_K\}_K$ to a given random variable $\rv{Y}$ is defined as $\lim\limits_{K\to \infty}\E [\rv{Y}_K] = \E [\rv{Y}]$. It implies the convergence in distribution but is weaker than the convergence \emph{in mean} $\lim\limits_{K\to \infty} \E[{|\rv{Y}_K-\rv{Y}|}^r] = 0, ~r \ge 1$.} to a random variable $\rv{Y}$ with CDF $F_\rv{Y}(y) = 1-e^{-y^\nt}$, i.e.,
  \begin{align}
    \lim\limits_{K\to \infty} \E \left[a_K{\displaystyle\min_{k\in[K]}} \normh\right] &= \E[\rv{Y}] = \Gamma
\left(1+\frac{1}{\nt}\right), \label{eq:conv-mean}\\
\hspace{-.7cm}\lim\limits_{K\to\infty} \E \left[\left(a_K{\displaystyle\min_{k\in[K]}} \normh\right)^2\right] &=
\E[\rv{Y}^2] = \Gamma \left(1+\frac{2}{\nt}\right). \label{eq:conv-mean2}
  \end{align}%
\end{lemma}

\begin{lemma} \label{lemma:SNR01}
  When $\nt$ grows at least logarithmically with $K$ such that $\ntgelnK$, we have
  \begin{align}
    \E \left[\min_{k\in[K]} \normh\right] = \Theta(1), \label{eq:tmp888} \\
    \E \left[\left(\min_{k\in[K]} \normh\right)^2\right] = \Theta(1). \label{eq:tmp889} 
  \end{align}%
  Further, if $\nt$ grows faster than $\ln(K)$ such that $\ln(K) = o(\nt)$, we have
  \begin{align}
    \min_{k\in[K]} \normh \xrightarrow{\text{p}} 1. \label{eq:fast}
  \end{align}%
\end{lemma}

\section{Scalable content delivery with wireless multicasting}
\label{sec:scalableMulticasting}

In this section, we focus on content delivery via wireless multicasting. Unlike in the original
works~\cite{maddah2013fundamental,maddah2013decentralized} on coded caching where the multicast link is perfect and has constant rate,
here the multicasting is performed over a multi-antenna wireless channel. Therefore, the multicast rate depends on the system parameters
such as the number of users and the number of transmit antennas. 
We summarize the main results of this section with a fixed transmit power in the following theorem.
\begin{theorem} \label{theo:scalableSchemes}
	Let us consider a content delivery system with a $\nt$-antenna base station 
	and $K$ single-antenna users. We assume no CSIT and coded caching is used with wireless multicasting. 
        Then, linear scaling of the content delivery rate w.r.t.~$K$ can be achieved with a fixed transmit power in the following cases:
	\begin{enumerate}
		\item with a large array of transmit antennas and/or when the multicast transmission can span over $L\ge 1$ independent sub-channels such that $\LntgelnK$;
		\item with a threshold-based user selection using one-bit feedbacks in a quasi-static fading channel, for an arbitrary number of transmit
                  antennas.
	\end{enumerate} 
\end{theorem}
In the rest of the section, we shall show the scalability of each case. For Case~1, we first investigate the extreme case $L = 1$ and provide some extra results for different setting of transmit power, then we extend to the general case $L\ge 1$ and prove the scalability with a fixed power.

\subsection{MISO multicasting in quasi-static channels ($L=1$)} \label{sec:quasistatic}
We first consider the case where all the $K$ users are served with MISO multicasting in a 
quasi-static Rayleigh fading channel ($L=1$). For simplicity, we assume that Gaussian signaling is used to send the multicast message~(also
called the common message), i.e., $\rvVec{X} = \rvVec{X}_0\sim\mathcal{CN}(0,\pmb{Q}_0)$ where $\pmb{Q}_0$ is the input covariance matrix. In this case, it follows that the maximum instantaneous multicast rate for a channel realization $\rvMat{H} = \pmb{H} = [\pmb{h}_1\ \cdots\ \pmb{h}_K]^\T$ is 
\begin{align}
R_0(\pmb{H}) = \underset{\pmb{Q}_0: \trace(\pmb{Q}_0)\le P}{\max} \min\limits_{k\in[K]} \ln(1+ \pmb{h}_k^\T \pmb{Q}_0 \pmb{h}_k^*).
\label{eq:R00}
\end{align}%
The input covariance matrix~$\pmb{Q}_0$ can be regarded as a precoding and spatial power allocation strategy. The inner
minimization in \eqref{eq:R00} is the achievable rate of the worst user for a given strategy $\pmb{Q}_0$, and is thus the maximum
multicast rate so that every user can decode the common message.\footnote{Alternatively, we can multicast at the rate of the worst user among those interested in decoding the message. For example, in centralized coded caching with $Km \eqdef t \in \mathbb{N}^+$, each coded packet is useful for a set $\Sc$ of $t+1$ users. Therefore, the packet can be transmitted at the rate of the worst user in $\Sc$, as considered in~\cite{Shariatpanahi2017multiantenna}. This improves the transmission rate when $\Sc$ does not contain the globally worst user. The occurrence rate of this event out of all possible sets of $t+1$ users is $1-\frac{t+1}{K} \xrightarrow{K\to\infty} 1-m$. Thus, in the large $K$ regime, this improvement is less significant when the user cache size grows.}
The outer maximization means that the transmitter can choose a strategy
that maximizes the multicast rate. Since we assume that the channel is not known at the transmitter and the channel is isotropic with
i.i.d.~Rayleigh fading, it is reasonable to use isotropic signaling $\rvVec{X}_0\sim\mathcal{CN}(0,\frac{P}{\nt}\Id_\nt)$, then $R_0(\pmb{H}) = \ln\bigl(1+\frac{P}{\nt} \min_{k\in[K]}\|\pmb{h}_k\|^2 \bigr)$. 
Let us define the SNR at user $k$ as $\SNR_k(\rvMat{H}) \defeq \frac{P}{\nt} \|\rvVec{H}_k\|^2$.
Then, the long-term average multicast rate is
\begin{align}\label{eq:defR0}
  \bar{R}_0 := \E [R_0] = \E\left[ \ln\biggl(1+ \minSNRo  \biggr) \right], 
\end{align}%
From \eqref{eq:Rmul} and \eqref{eq:defR0}, the equivalent content delivery rate is
\begin{align}
\Rmc = \frac{K}{T(m,K)}\E\left[ \ln\left(1+ \minSNRo  \right) \right].
\end{align}%



	\begin{proposition} \label{prop:R0}

          In the large $K$ regime, the asymptotic behavior of the long-term average multicast rate depends on the size of the transmit
          antenna array, as described in Table~\ref{tab:R0}.  
	\end{proposition}
        \begin{table*}[t!]
          \caption{Asymptotic behavior of the average multicast rate when $K\to\infty$}
          \centering
          \begin{tabular}{|>{\centering\arraybackslash}p{0.12\textwidth}||c||>{\centering\arraybackslash}p{0.12\textwidth}||c|}
            \hline
            \multicolumn{2}{|>{\centering\arraybackslash}p{0.4\textwidth}||}{small antenna array$^*$: $\nt = \Theta(1)$} &
            \multicolumn{2}{>{\centering\arraybackslash}p{0.4\textwidth}|}{large antenna array: $\ntgelnK$} \\ \hline 
            \hline
            $P = o(K^{\frac{1}{\nt}})$ & $\bar{R}_0 \approxasympt \frac{P}{a_K} \Gamma(1+\frac{1}{\nt}) = o(1)$
            & $P = o(1)$ & $\bar{R}_0 = \Theta(P)  = o(1)$  \\
            \hline
            $P = \Theta(K^{\frac{1}{\nt}})$ & $\bar{R}_0 = \Theta(1)$ & $P = \Theta(1)$ & $\bar{R}_0 = \Theta(1)$ \\
            \hline
            $P K^{-\frac{1}{\nt}} \to \infty$ 
             & $\bar{R}_0 = \Theta\left(\ln\left(1+\frac{P}{a_K}
            \Gamma(1+\frac{1}{\nt})\right)\right)$ & $P\to\infty$ & $\bar{R}_0 = \Theta\left(\ln(1+P)\right)$  \\
            \hline
          \end{tabular}
        \label{tab:R0} \\
        \vspace{-.15cm}
        \begin{flushleft}
		\hspace{1.1cm}\footnotesize {$^*$Recall that $a_K \defeq \nt\left(\frac{K}{\nt!}\right)^{1/\nt}$}.
		\end{flushleft}
        \end{table*}
        Before proving the proposition, some comments on the asymptotic results are in place.  In the small antenna array regime where the
        $\nt$ does not scale up with $K$, the multicast rate vanishes when $K\to\infty$ if the total transmit power scales with the number of users slower than $K^{\frac{1}{\nt}}$, i.e., $P = o(K^{\frac{1}{\nt}})$. If $P$ increases with $K$ as fast as $K^{\frac{1}{\nt}}$, a fixed multicast rate can be maintained. Further, if $P$ increases with
        $K$ faster than $K^{\frac{1}{\nt}}$, the multicast rate can also grow with $K$. Intuitively, for a fixed number of transmit antennas, the channel
        quality of the worst user degrades with the total number of users. A remedy for this is to increase the transmit power with $K$,
        which is however not desirable (if not impossible) in many practical situations. Another solution is to increase the number of transmit antennas with $K$. According to the right-hand side of Table~\ref{tab:R0}, in the large antenna array regime where $\nt$ is asymptotically larger than $\ln(K)$, a
        constant amount of transmit power suffices to maintain the non-vanishing multicast rate. The interpretation behind this is the \emph{channel
        hardening} effect that decreases the variance of the individual SNR with $K$ so that the worst user can still have a constant rate.\footnote{Note that the rate scaling in Table~\ref{tab:R0} agrees with the capacity scaling derived in~\cite{jindal2006capacity} for the case of a fixed total power. While~\cite{jindal2006capacity} proves that the multicast capacity is non-vanishing when the number of antennas scales \textit{linearly} with the number of users, we relax this condition by showing that a \textit{logarithmic} scaling is sufficient.}

        \begin{remark}
        Interestingly, to see the sufficiency of the logarithmic scaling of $\nt$, a heuristic way is to let $P$ grow in the small array
        regime to maintain the multicast rate, i.e., let $P = K^{\frac{1}{\nt}}$ as suggested above.  Now we see that if $\nt = \ln K$,
        then $P=K^{\frac{1}{\ln K}}\to e$ which is bounded. In general, it is enough to have that the product $PK^{-\frac{1}{\nt}}$ is non-vanishing.
        \end{remark}
        
        We provide a formal proof of Proposition~\ref{prop:R0} in the following.


	\begin{proof}
          Essentially, the proof relies on Lemma~\ref{lemma:log-asympt}, according to the
          asymptotic behavior of $\E[\displaystyle\min_{k\in[K]} \SNR_k ]$. For convenience, let us define
          $\rv{S}_K := \displaystyle\min_{k\in[K]} \SNR_k = P\displaystyle\min_{k\in[K]} \normh$. 
          
          First, we consider the case of small antenna array with
          $\nt = \Theta(1)$. From \eqref{eq:conv-mean}, we have 
          \begin{align}
            \SNRmean^{-1} {\E[\rv{S}_K]} \to 1,
          \end{align}%
          where $\SNRmean := {\frac{P}{a_K} \Gamma\left(1+\frac{1}{\nt}\right)} = \Theta\bigl(P K^{-\frac{1}{\nt}}\bigr)$, since
          $a_K = \Theta\bigl( K^{\frac{1}{\nt}} \bigr)$. When $P = \Theta(K^{\frac{1}{\nt}})$, $\E[\rv{S}_K] =
          \Theta(1)$, and from case~1 of Lemma~\ref{lemma:log-asympt}, we have $\bar{R}_0 = \Theta(1)$. When 
          $P = o(K^{\frac{1}{\nt}})$, we have $\E[\rv{S}_K] = o(1)$. Since $\E[\rv{S}_K^2] = \frac{P^2}{
          a_K^2} \E[\left(a_K\displaystyle\min_{k\in[K]} \left\{ \normh \right\}\right)^2 ]$ which is
          $\Theta(P^2 K^{-\frac{2}{\nt}})$ according to \eqref{eq:conv-mean2}, we obtain $\E[\rv{S}_K^2] = o(\E[\rv{S}_K])$,
          and, from case~2 of Lemma~\ref{lemma:log-asympt}, we have $\bar{R}_0 \approxasympt \E[\rv{S}_K] \approxasympt
          \overline{\SNR}$.   
          When $P K^{-\frac{1}{\nt}} \to \infty$, we have 
          \begin{align}
          \Pr\left(\rv{S}_K \ge \SNRmean \right) &= \Pr\left(a_K \min_{k\in[K]} \normh \ge \Gamma\left(1+\frac{1}{\nt}\right)\right) \\
          &= \exp\left(-\left[\Gamma\left(1+\frac{1}{\nt}\right)\right]^\nt\right) + o(1),
          \end{align}
          which is bounded away from zero since $\nt$ is fixed, where the last equality is due to Lemma~\ref{lemma:SNR0}. 
We just verified that the condition required in case~3 of Lemma~\ref{lemma:log-asympt} is also met (with $g(K) = \SNRmean$), thus $\bar{R}_0 
= \Theta\left(\ln\left(1+\overline{\SNR}\right)\right)$.
          
          Next, let us consider the case of large antenna array with $\ntgelnK$. From \eqref{eq:tmp888}, we have
          $\E[\rv{S}_K] = \Theta(P)$. 
          The case $P = \Theta(1)$ follows readily from case~1 of Lemma~\ref{lemma:log-asympt}. 
          When $P = o(1)$, we have $\E[\rv{S}_K] = o(1)$. We also have $\E[\rv{S}_K^2] = \Theta(P^2)$ according to
          \eqref{eq:tmp889}, and thus $\E[\rv{S}_K^2] = o(\E[\rv{S}_K])$. From case~2 of Lemma~\ref{lemma:log-asympt}, we
          have $\bar{R}_0 \approxasympt \E[\rv{S}_K] = \Theta(P)$. 
          When $P\to\infty$, we use Lemma~\ref{lemma:CDFgamma} to have: \begin{align}
          \Pr\left(\rv{S_K} \ge \eta P\right) &= \left[\Pr\left(\normh \ge \eta \right)\right]^K \\
          &\ge (1 - e^{-\nt})^K \\&\ge (1 - e^{-\ln(K)-c})^K \\&= e^{-e^{-c}} + o(1),
          \end{align}
          for some $c > -\infty$.
          We just verified that the condition required in case~3 of Lemma~\ref{lemma:log-asympt} is also met (with $g(K) = \eta P$), thus $\bar{R}_0 
          = \Theta\left(\ln\left(1+\eta P\right)\right) = \Theta\left(\ln\left(1+P\right)\right)$.
	\end{proof}

\subsection{Multicasting over $L$-parallel channel}
\label{subsec:blockFading}


In this subsection, we consider the general $L$-parallel channel model when a codeword can span over $L\ge 1$ interference-free sub-channels, such that in each sub-channel, the signal is perturbed by independent fading coefficients and independent noise. 
It includes the block fading and multi-carrier systems, such as OFDM, as special cases, where the sub-channels correspond to coherence intervals and sub-carriers, respectively.

With isotropic signaling, the instantaneous multicast rate for a given realization $(\pmb{H}_1,\dots,\pmb{H}_L)$ of $L$ sub-channels is 
\begin{align}
R_0(\pmb{H}_1,\dots,\pmb{H}_L) = \min_{k\in[K]}\frac{1}{L}\sum_{l=1}^L\ln\biggl(1+\frac{P}{\nt} \|\pmb{h}_{k,l}\|^2  \biggr),
\end{align} 
where $\pmb{h}_{k,l}$ is the channel realization of user $k$ in sub-channel $l$. The SNR at user $k$ is now defined for each sub-channel $l$ as $\SNR_{k,l}(\rvMat{H}_l) \defeq \frac{P}{\nt} \|\rvVec{H}_{k,l}\|^2$ and the average multicast rate is
\begin{align}
  \bar{R}_0 = \E\left[\min_{k\in[K]}\frac{1}{L}\sum_{l=1}^L \ln\biggl(1+ \SNR_{k,l} \biggr) \right]. \label{eq:R0_L}
\end{align}%
The equivalent content delivery rate is given by plugging this multicast rate into \eqref{eq:Rmul}. Intuitively, when the number of
sub-channels $L$ grows to infinity fast enough w.r.t.~the number of users $K$, each user should have a constant rate and the multicast rate is
non-vanishing with $K$. 
Our goal is to find out the sufficient scaling of $L$ to guarantee a non-vanishing multicast rate. In the following, we focus on the case with a
constant power~$P$, i.e., $P=\Theta(1)$ when $K\to\infty$. Since the direct analysis of the rate~\eqref{eq:R0_L} is non-trivial, we
resort to the analysis of upper and lower bounds of this rate. Let us define $\SNR_{j,k,l}:=P |\rv{H}_{j,k,l}|^2$ where
$\rv{H}_{j,k,l}$ is the channel coefficient from the $j$-th transmit antenna to the $k$-th user in the $l$-th sub-channel. Then, we can write $\SNR_{k,l}= \frac{1}{\nt} \sum_{j=1}^{\nt} \SNR_{j,k,l}$, $\forall\,k,l$. From the concavity of the
logarithm function, we have the following upper and lower bounds:
\begin{align}
  \bar{R}_0 &\le \E[\ln\left( 1 + \min_{k\in[K]} \frac{1}{L\nt} \sum_{l=1}^L \sum_{j=1}^{\nt} \SNR_{j,k,l}\right)]
  \label{eq:R0L_ub}, \\
  \bar{R}_0 &\ge \E[\min_{k\in[K]} \frac{1}{L\nt} \sum_{l=1}^L\sum_{j=1}^{\nt}\ln\left( 1 +  \SNR_{j,k,l}\right)].
  \label{eq:R0L_lb} 
\end{align}%
It turns out that the above bounds are enough to establish the sufficient scaling of both $L$ and $\nt$ needed to maintain a non-vanishing multicast rate. 
\begin{proposition} \label{prop:R0blockfading}
  If $\LntgelnK$ and $P$ is fixed, then $\bar{R}_0 = \Theta(1)$ when $K\to\infty$.
\end{proposition}
The above result demonstrates an interesting trade-off between the number of transmit antennas and the number of independent sub-channels for a scalable
multicast rate. A large number of sub-channels can compensate for the limited number of transmit antennas, and vice versa.   
\begin{remark}
Since $\nt$ and $L$ are respectively the spatial and temporal/frequency diversity per user, the product $L\nt$ can be interpreted as the \emph{total
diversity} that can be exploited by each user. Proposition~\ref{prop:R0blockfading} says that as long as the {\em total diversity} is asymptotically
larger than $\ln(K)$, the multicast rate is not vanishing. As shown in subsection~\ref{sec:quasistatic}, an extreme case is $L=1$, in which a large array of antennas should be exploited to obtain scalability. The other extreme case is $\nt = 1$, in which the server uses single antenna, as assumed in~\cite{maddah2013fundamental,maddah2013decentralized}, but spans the multicast transmission over $L$ independent sub-channels with $L$ asymptotically larger than $\ln(K)$.
\end{remark}
\begin{proof}[Proof of Proposition~\ref{prop:R0blockfading}]
  Following Proposition~\ref{prop:R0}, we can readily show that the upper bound~\eqref{eq:R0L_ub} is
  $\Theta(1)$ when $\LntgelnK$. This is because, due to the i.i.d.~property
  across both blocks and antennas, the upper bound is exactly the same as \eqref{eq:defR0} if we replace $\nt$ by $L\nt$.
  We can therefore focus on the lower bound~\eqref{eq:R0L_lb}. Let us consider the following CDF
  \begin{align}
F(r) := \Pr\left(\frac{1}{L\nt} \sum_{l=1}^L\sum_{j=1}^{\nt}\ln\left( 1 +
\SNR_{j,k,l}\right) \le r\right). 
  \end{align}%
  Using the Chernoff bound~\eqref{eq:Chernoff}, we have, for any $\nu>0$, 
  \begin{align}
    F(r) &\le e^{L\nt \nu r } \E[(1+\SNR_{j,k,l})^{-\nu}]^{L\nt}\\
    &= \Biggl(\frac{e^{-\nu r }}{
    \E[(1+\SNR_{j,k,l})^{-\nu}]}\Biggr)^{-L\nt} \\
    &\le g(\nu,r)^{-\ln K - c_K},  \label{eq:tmp222}
  \end{align}%
  where we define $g(\nu,r):=\frac{e^{-\nu r }}{
    \E[(1+\SNR_{j,k,l})^{-\nu}]} = \frac{\exp\left(-\nu r -\frac{1}{P}\right)}{
    \Gamma\left(1-\nu,\frac{1}{P}\right)}$; in the last inequality $c_K = O(1)$ from the assumption that
  $\LntgelnK$. 
  It can be verified that there exist $\nu_0 = \Theta(1)$ and $r_0 = \Theta(1)$ such that $g(\nu_0,r_0) = e$. Therefore, from \eqref{eq:tmp222} we
  obtain $F(r_0) \le \frac{e^{-c_K}}{K}$. Now, by applying
  \eqref{eq:Markov1} on \eqref{eq:R0L_lb},
  \begin{align}
    \bar{R}_0 &\ge r_0 (1-F(r_0))^K \\
    &\ge r_0 \Bigl(1-\frac{e^{-c_K}}{K}\Bigr)^{K} \\
    &= r_0(e^{-e^{-c_K}}+o(1)), \quad \text{when $K$ is large}, 
  \end{align}%
  which is $\Theta(1)$ since $c_K = O(1)$. 
\end{proof}


\subsection{Multicasting with user selection}
\label{subsec:SISO}

Since the bottleneck of multicast transmission is the channel quality of the worst users, the transmission rate can be improved if we
only serve users with better quality. In other words, we eliminate users with ``\emph{unacceptable}'' channel qualities. 
For instance, if we transmit at the average~(median) rate over the channel gain, then the number
of users being able to decode is roughly $K/2$, and we can guarantee a linear sum rate scaling. The trade-off between the multicast rate
and number of users served should be balanced so as to maximize the sum rate. 
In order to achieve linear scaling with the total number of users, a non-negligible fraction of the $K$ users should be selected. 
In this work, we propose a threshold-based user selection scheme.

Here is how the scheme works. Let us first focus on the single transmit antenna and quasi-static fading case, i.e., $\nt= L=1$.
We assume that the base station fixes a SNR threshold $s$ and reveals it to all the users prior to the actual data transmission. 
Then, each user sends back an one-bit feedback indicating whether the instantaneous received SNR is above the threshold. 
Let the random variable $\rv{K}^*(s)$ be the number of users with SNR above the threshold, i.e., $\rv{K}^*(s) \defeq |\{k: \SNR_k \ge
s\}|$.  Recall that $\SNR_1, \dots, \SNR_K$ are $K$ i.i.d.~exponential random variables ${\rm Exp}(\frac{1}{P})$, then $\E[\rv{K}^*(s)] = K\Pr(\SNR \ge s)=
Ke^{-s/P}$. The base station then starts the multicast transmission at rate $\ln(1+s)$
so that every selected user is able to decode the common message. 

Since the set of active users changes frequently under user selection, it is more reasonable to assume that decentralized
placement~\cite{maddah2013decentralized} is used. From \eqref{eq:defT} and \eqref{eq:Rmul}, we have the equivalent content delivery rate 
\begin{align}
\Rmc &= \E \left[\frac{\frac{m}{1-m}\rv{K}^*(s)}{1-(1-m)^{\rv{K}^*(s)}} \ln(1+s)\right].
\end{align}
From the strong law of large numbers, we know that $\frac{\rv{K}^*(s)}{K} \xrightarrow{\text{a.s.}} \Pr(\SNR \ge s) = e^{-s/P}$, which means that 
\begin{align}
\frac{\frac{m}{1-m}\frac{\rv{K}^*(s)}{K}}{1-(1-m)^{\rv{K}^*(s)}} \ln(1+s)\xrightarrow{\text{a.s.}}\frac{m}{1-m} e^{-s/P} \ln(1+s).
\end{align} 
Therefore,
using the dominated convergence theorem, we obtain
\begin{align}
\frac{\Rmc}{K} &\approxasympt \frac{m}{1-m} e^{-s/P} \ln(1+s), 
\end{align}
which shows that for {\em any non-zero threshold $s$}, linear scaling can be achieved. In practice, however, it is desirable 
to find a threshold that maximizes the scaling factor $e^{-s/P} \ln(1+s)$. 
Since this factor is zero when $s=0$ and $s=\infty$, due to the continuity, $e^{-s/P} \ln(1+s)$ is maximized by some $0<s^*<\infty$ that
satisfies $\frac{\text{d} \left(e^{-s/P} \ln(1+s) \right)}{\text{d} s} \Bigr|_{s=s^*} = 0$. It follows that  
\begin{align} \label{eq:optimalS_cond}
\ln(1+s^*) = W\left(P\right),
\end{align}
where $W(\cdot)$ is the Lambert-W function such that $W(x) e^{W(x)} = x$.  
Therefore, when $K$ is large, we should choose a SNR threshold 
\begin{align} \label{eq:optimalS}
s^* = e^{W(P)} - 1 = \frac{P}{W(P)}-1.
\end{align}
The corresponding optimal content delivery rate is 
\begin{align} \label{eq:optRateScheduling}
\Rmc \approxasympt \frac{m}{1-m} K e^{\frac{1}{P}-\frac{1}{W(P)}}W(P),
\end{align}
scaling linearly with $K$. The expected number of selected users is $K^*(s^*) = K e^{\frac{1}{P}-\frac{1}{W(P)}}$. In~\cite{Asma2017opportunistic}, it is shown that the rate~\eqref{eq:optRateScheduling} is indeed the optimal rate of any scheme that transmits opportunistically to the group of users with the highest sum content delivery rate at each channel realization. 

The above result can be readily extended to any i.i.d.~SNR distribution with differentiable CDF $F_{\SNR}(s)$, e.g., the case with
multiple transmit antennas. Specifically, the optimal SNR threshold $0<s^*<\infty$ should satisfy $\frac{\text{d} \left( (1-F_{\SNR}(s))
	\ln(1+s) \right)}{\text{d} s} \Bigr|_{s=s^*} = 0$. We readily obtain the following result. 
\begin{proposition} 
	Let us consider a multicast channel with i.i.d.~SNR distribution with differentiable CDF $F_{\SNR}(s)$. 
	Define $f(s):=\frac{1-F_{\SNR}(s)}{F'_{\SNR}(s)}$ for $s>0$, then 
	the optimal SNR threshold $s^*$ for user selection is such that
	\begin{align} 
	\ln(1+s^*) = W\left(f(s^*)\right).
	\end{align}
	The corresponding optimal content delivery rate is 
	\begin{align}
	\Rmc &\approxasympt \frac{m}{1-m} K (1-F_{\SNR}(s^*)) W(f(s^*)).
	\end{align}%
\end{proposition}
In general, an explicit expression of the optimal threshold is hard to derive. Nevertheless, such value can be 
obtained numerically. 

\subsection{Numerical results}
To validate Theorem~\ref{theo:scalableSchemes}, we calculate numerically the equivalent content delivery rate $\Rmc$ and
observe its behavior when $K$ increases. In Fig.~\ref{fig:Rmul}, we plot the rate achieved with the scalable schemes listed in
Theorem~\ref{theo:scalableSchemes} as a function of the number of users $K$ for normalized cache size $m = 5\%$ and a fixed total power. Specifically, we consider multicasting with 1) $\nt = \lfloor\ln(K)\rfloor$ antennas, 2) single antenna and transmission spanning over $L = \lfloor\ln(K)\rfloor$ sub-channels, and 3) single antenna and threshold-based user selection scheduling (Case 1 and Case 2 are the extreme cases of $L\nt = \lfloor\ln(K)\rfloor$). It can be seen clearly that the sum rate scales linearly with $K$ in these cases. For a baseline, we also plot the rate achieved with single-antenna and without user selection in quasi-static fading channel. In this case, the sum rate saturates when $K$ is large and hence the system is not scalable.

Another observation from Fig.~\ref{fig:Rmul} is that higher total power yields higher content delivery rate. This gain due to power is more pronounced in the baseline case $\nt = L = 1$, no scheduling. In this case, from Table~\ref{tab:R0}, the multicast rate scales as $\bar{R}_0 \sim \frac{P}{K}$, and hence the content delivery rate scales as $\Rmc = \frac{K}{T(m,K)} \bar{R}_0 \sim \frac{mP}{1-m}$, linearly in $P$, but constantly w.r.t.~$K$.
\begin{figure}[!h]
	\centering
	\subfigure[$P = 30$~dB]{
		\includegraphics[width=.5\textwidth]{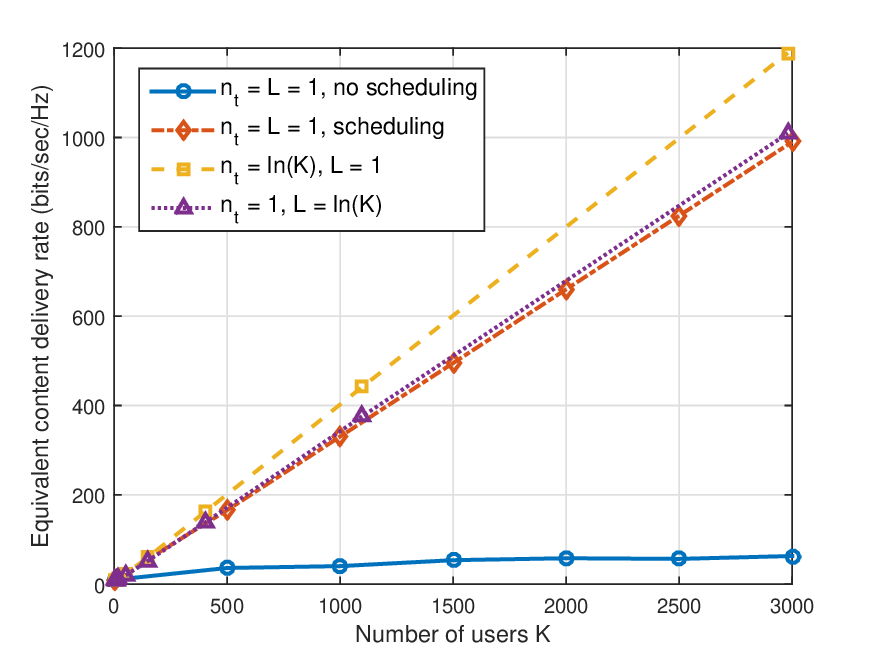}}
	\subfigure[$P = 40$~dB]{
		\includegraphics[width=.5\textwidth]{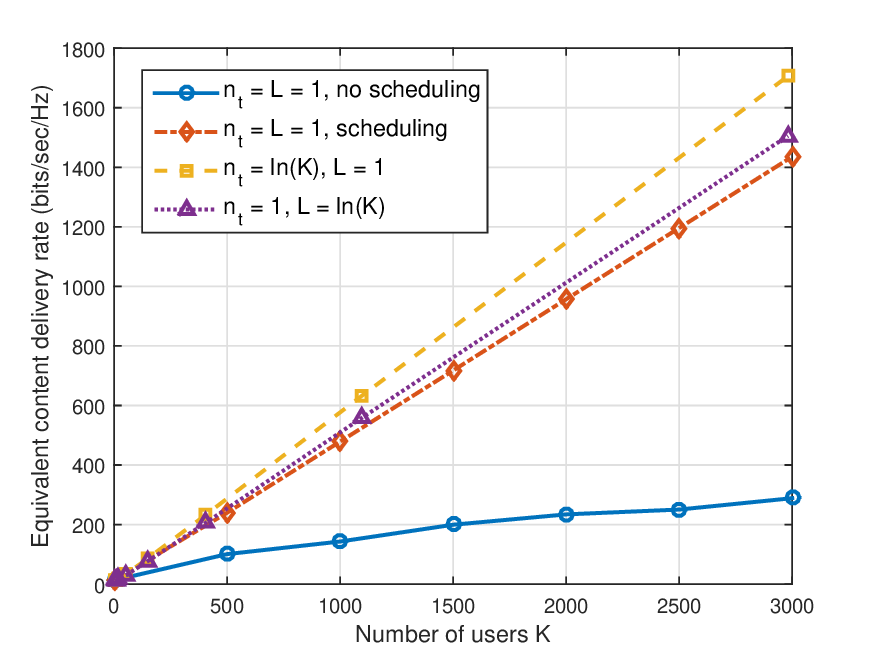}}
	\caption{The equivalent content delivery rate achieved with different multicasting schemes, namely, 1) $\nt=1$
		without scheduling in quasi-static fading, 2) $\nt = 1$ with scheduling in quasi-static fading, 3) $\nt = \lfloor\ln(K)\rfloor$ in quasi-static fading, 4) $\nt = 1$ and transmit over $L = \lfloor\ln(K)\rfloor$ independent sub-channels, as a function of $K$ for $m = 5\%$ and $P = 30, 40$~dB.}
	\label{fig:Rmul}
\end{figure}

In Fig.~\ref{fig:scheduling}, we plot the asymptotically optimal SNR threshold $s^*$ given in~\eqref{eq:optimalS} for user selection and the exact optimal solution from simulation. We observe that the analytical solution converges to the exact optimal one when $K$ is large. Since the analytical optimal SNR threshold~\eqref{eq:optimalS} only depends on the total power, it can be predefined easily by the base station. 
\begin{figure}[!h]
	\centering
	\includegraphics[width=.5\textwidth]{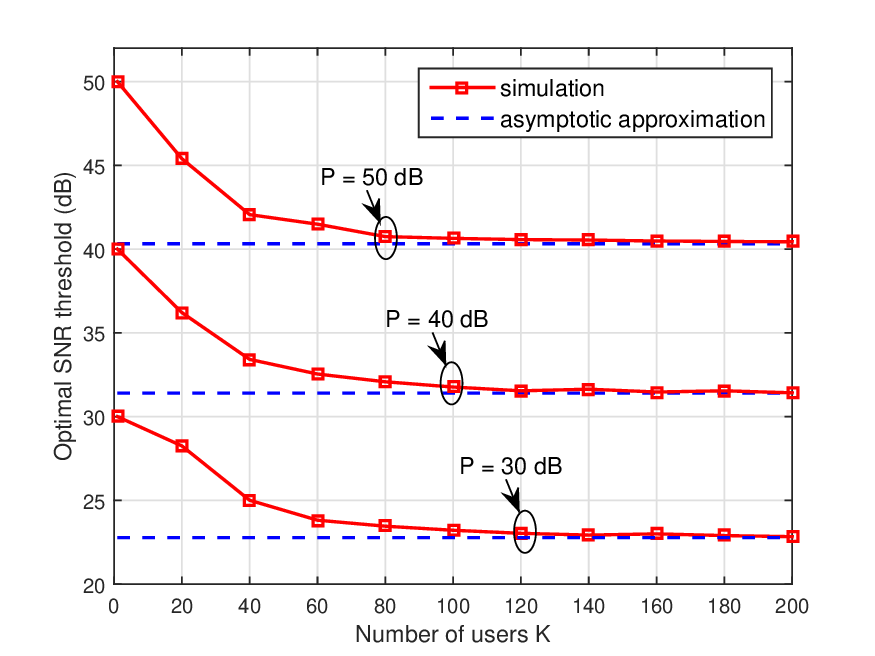}
	\caption{The optimal SNR threshold for user selection scheduling with $m = 5\%$, $P = 30, 40, 50$~dB and single antenna: asymptotic approximation vs. simulation.}
	\label{fig:scheduling}
\end{figure}

\section{Scalable content delivery with spatial multiplexing with CSIT}
\label{sec:scalableMultiplexing}
\newcommand{\Rsym}{\bar{R}_{\rm sym}}
\newcommand{\hest}{\hat{\rvVec{h}}}
\newcommand{\Hest}{\hat{\rvVec{H}}}
\newcommand{\SINRzf}{\SINR_k}
\newcommand{\SINRmf}{\SINR_k^{\rm {mrc}}}
\newcommand{\Rsymzf}{\Rsym^{\rm {zf}}}
\newcommand{\Rsymmf}{\Rsym^{\rm {mf}}}
\renewcommand{\diag}[1]{{\rm diag}\left(#1\right)}

Instead of using wireless multicasting and coded caching, a more conventional content delivery scheme is spatial
multiplexing. Specifically, simultaneous unicast transmissions can be realized with spatial precoding based on the available CSIT. With
spatial multiplexing, the required content is delivered directly to the user. In this section, we study the content delivery
rate of this scheme. For simplicity, we consider the quasi-static fading channels ($L=1$) in the rest of the paper.

With linear precoding, the transmitted signal is 
\begin{align}
\rvVec{X} = \sum_{k=1}^K \rvVec{W}_k \rv{X}_k,
\end{align}
where for user
$k\in[K]$, $\rv{X}_k$ is the {\it private} signal and $\rvVec{W}_k$ is the precoder of unit norm that depends only on the estimated
channel matrix~$\hat{\rvMat{H}}$ as defined in \eqref{eq:imperfectCSIT}. In this
work, we assume that $\nt \ge K$ and focus on zero-forcing~(ZF) precoder. The precoding vector $\{\rvVec{W}_k\}$ for user~$k$ is
\begin{align}
	\rvVec{W}_k &= \alpha_k \rvMat{U}_k \rvMat{U}_k^\H \hat{\rvVec{H}}_k^*,
\end{align}%
where the columns of $\rvMat{U}_k$ form an orthonormal basis of the null space of
$\text{span}(\{\hat{\rvVec{H}}^*_l\}_{l\ne k})$ and $\rvMat{U}_k$ is assumed to be independent of
$\hat{\rvVec{H}}_k$; $\alpha_k \defeq \frac{1}{\|\hat{\rvVec{H}}_k^\T \rvMat{U}_k\|}$ is the normalization
factor such that $\|\rvVec{W}_k\| = 1$. Intuitively, we project the signal of user $k$ onto the null space
of all other users' channels to eliminate the interference and then align with its own channel to maximize
the received signal power. Note that each precoding vector here is normalized so that each stream can
have the same power.  We use i.i.d.~Gaussian signaling for tractability, i.e., $\{\rv{X}_k\}$ are
i.i.d.~$\mathcal{CN}(0, P_k)$ with sum power constraint $\sum_{k=1}^{K}P_k = P$. User $k$ receives the signal
\begin{align} \label{eq:ZFrxsignal}
	\rv{Y}_k 
	&= \rv{G}_k \rv{X}_k + \sum_{l\ne k} \tilde{\rv{G}}_{k,l} \rv{X}_l + \rv{Z}_k,
\end{align}%
where $\rv{G}_{k} \defeq {\rvVec{H}}_k^\T \rvVec{W}_k$ and $\tilde{\rv{G}}_{k,l} \defeq \tilde{\rvVec{H}}_k^\T \rvVec{W}_l \sim
\mathcal{CN}(0,	\sigma^2)$. It is worth mentioning that the above equivalent channel coefficients are not independent between each other. Let us
define the signal-to-interference-plus-noise ratio~(SINR) at receiver~$k\in[K]$ as
\begin{align}
	\SINRzf(\rvMat{H}) &\defeq
	\frac{{|\rv{G}_k|}^2 P_k}{1+ \sum_{l\ne k} |\tilde{\rv{G}}_{k,l}|^2 P_l }. \label{eq:SINRzf}
\end{align}%
For any realization $\rvMat{H} = \pmb{H}$, we obtain the instantaneous rate $R_k(\pmb{H}) = \ln \left(1 + \SINR_k(\pmb{H}) \right)$ for user~$k\in [K]$.
The long-term average unicast rate of user~$k$ is
\begin{align}
\bar{R}_k &:= \mathbb{E}\left[\ln \left(1 + \SINR_k(\rvMat{H}) \right) \right].
\end{align}%
For simplicity, we consider uniform power allocation, i.e., $P_k=\frac{P}{K} =:p,\forall\,k\in[K]$. Then $\SINRzf = \frac{{|\rv{G}_k|}^2}{p^{-1}+
\sum_{l\ne k} |\tilde{\rv{G}}_{k,l}|^2}$. Due to the symmetry of the problem, the marginal distribution of $\SINRzf$ does not depend on $k$, and
can be described as follows.
\newcommand{\SINRsym}{\SINR_{\rm sym}}
\begin{lemma} \label{lemma:SINRzf}
	With uniform power allocation~($p=P/K$), $\SINR_k$ can be written, in distribution, as 
        \begin{align}
          \SINR_k &\overset{d}{=} \SINRsym \defeq \frac{\left| \sigma \rv{A}_K + \sqrt{(\nt-K+1)(1-\sigma^2)
          \rv{B}_K}
          \right|^2}{p^{-1} + (K-1) \sigma^2 \rv{C}_K }, \label{eq:SINRsym}
        \end{align}%
        for some joint distribution of~$(\rv{A}_K, \rv{B}_K, \rv{C}_K)$
        such that $\rv{A}_K\sim\mathcal{CN}(0,1)$ and $\rv{B}_K \sim \text{\rm Gamma}(\nt-K+1,\frac{1}{\nt-K+1})$
        are independent, and $\E[\rv{C}_K] = 1$. In addition, when $\liminf\limits_{K\to\infty}\frac{\nt}{K} > 1$, we
        have 
        \begin{align}
          \rv{B}_K &\xrightarrow{\text{a.s.}} 1 \quad \text{and} \quad
          \rv{C}_K \xrightarrow{\text{a.s.}} 1. \label{eq:tmpas} 
        \end{align}%
\end{lemma} 
\begin{proof}
	The proof is provided in Appendix~\ref{proof:lemmaSINRzf}.
\end{proof}
In this work, we focus exclusively on the case with $\liminf\limits_{K\to\infty}\frac{\nt}{K} > 1$ to
gain some insight on the behavior of ZF precoding.
The case with $\nt = K$ is too involved\footnote{When $\nt = K$, the almost
sure convergence of the sum $\frac{1}{K-1} \sum_{l\ne k} |\tilde{\rv{G}}_{k,l}|^2$ does not hold. We need
to establish upper and lower bounds to derive the scaling of $\Rsym$.} for our purposes here and is not considered. With uniform power allocation, the long-term unicast rate is also symmetric, i.e., $\bar{R}_k = \Rsym, ~\forall\,k\in[K]$.
Using Lemma~\ref{lemma:SINRzf}, we can derive the asymptotic behavior of $\Rsym$ in the large
$K$ regime. 
\begin{proposition} \label{prop:Rsym}
With uniform power allocation ($p = P/K$) and $\liminf\limits_{K\to\infty}\frac{\nt}{K} > 1$, we have 
\begin{align}
  \Rsym &\approxasympt \begin{cases}
    \frac{1 + (\nt-K+1)(1-\sigma^2)}{p^{-1} + K-1 }, \\
    \hspace{1cm}\text{when {\small$(\nt-K+1)(1-\sigma^2) =   O(1)$}},\\
\ln\left( 1 +  \frac{(\nt-K+1)(1-\sigma^2) }{p^{-1} + (K-1) \sigma^2 } \right), \\
\hspace{1cm}\text{when {\small$(\nt-K+1)(1-\sigma^2) \to \infty$}}. \\
  \end{cases}
  \label{eq:Rsym}
\end{align}%
\end{proposition}
Before proving the proposition, we provide some observations. 
The asymptotic behavior of $\Rsym$ depends on the channel estimation error $\sigma^2$. If the channel estimation fails when $K\to\infty$ in such a
way\footnote{This can happen when the resources for channel estimation saturate with a large number of users.}
that $(\nt-K+1)(1-\sigma^2) = O(1)$, we see from \eqref{eq:Rsym} that the symmetric rate decays with $K$ as $1/K$ for given total power $P$. Otherwise, the
symmetric rate depends on $(\nt,K,p,\sigma^2)$ in a non-trivial way. The case of particular interest is when the estimation error variance $\sigma^2$ is
fixed and strictly smaller than $1$, in this case the symmetric rate does not vanish with $K$ for fixed total power $P$. Indeed, according to
\eqref{eq:Rsym}, $\Rsym$ can even grow unboundedly with $\frac{\nt}{K}$ thanks to the \emph{beamforming gain}. We shall have more discussion on this
assumption at the end of this section. 
\begin{proof}[Proof of Proposition~\ref{prop:Rsym}]
When $(\nt-K+1)(1-\sigma^2) = O(1)$, we have $1-\sigma^2 \to 0$ since $\nt-K+1 \to \infty$. 
From \eqref{eq:SINRsym}, we notice
that $\SINRsym\xrightarrow{\text{a.s.}}0$. Thus, $\ln(1+\SINRsym) \approxasympt \SINRsym $ when
$K$ is large, and $\Rsym = \E[\ln(1+\SINRsym)]\approxasympt \E[\SINRsym] $ becomes
\begin{align}
  \Rsym 
  &\approxasympt \frac{\sigma^2 + (\nt-K+1)(1-\sigma^2)}{p^{-1} + (K-1) \sigma^2 } \\
  &\approxasympt \frac{1 + (\nt-K+1)(1-\sigma^2)}{p^{-1} + K-1 }. 
\end{align}%
When $(\nt-K+1)(1-\sigma^2) \to \infty$, 
from \eqref{eq:SINRsym} and \eqref{eq:tmpas}, it follows that 
\begin{align}
 \SINRsym \frac{p^{-1} + (K-1) \sigma^2 }{{(\nt-K+1)(1-\sigma^2) }
 }  &\xrightarrow{\text{a.s.}} 1
\end{align}%
and thus
  $\Rsym \approxasympt \ln\left( 1 +  \frac{(\nt-K+1)(1-\sigma^2) }{p^{-1} + (K-1) \sigma^2 } \right)$. 
\end{proof}

For content delivery, since we assume that each user already caches in average a fraction $m$ of the requested file, 
to \emph{complete} the file of $F$ bits, the base station needs to send $(1-\normmem) F$ bits. With spatial multiplexing, this transmission takes
$(1-\normmem) F / \Rsym$ units of time in average. It follows that the equivalent content delivery rate of the system is simply
\begin{align} \label{eq:Runi}
\Ruc = \frac{K}{1-\normmem}\Rsym.
\end{align}%
\begin{example}
  Let us consider a commonly used, albeit simplified, MMSE channel estimation model with $\sigma^2 =
  \frac{1}{1+p}$. Then, it follows that 
  \begin{align}
  \begin{cases}
  \sigma^2 = \Theta(p^{-1}), &\text{when $p\to\infty$}, \\
  1-\sigma^2 = \Theta(p), &\text{when $p\to 0$}, \\
  \sigma^2 = \Theta(1), 1-\sigma^2 = \Theta(1), &\text{when $p$ is fixed}.
  \end{cases}
  \end{align}
  Further, we assume that $\lim\limits_{K\to\infty}\frac{\nt}{K} = \beta > 1$. From \eqref{eq:Rsym}, on one hand, we see that
  if the per-user power $p \to 0$ when $K\to \infty$, the symmetric transmission rate vanishes
  as $\Rsym = (\beta-1)\Theta(p)$, thus $\Ruc = (\beta-1)\Theta(K p)$. On the other hand, if the per-user power is not
  vanishing with $K$, i.e., $p = \Omega(1)$, then $\Rsym = \Omega(1)$ and thus $\Ruc = \Omega(K)$. 
\end{example}

\begin{remark}
The above example shows that, when CSIT error is inversely proportional to the per-user power $p$, content delivery with spatial multiplexing
requires at least a linearly increasing total transmit power ($P = \Omega(K)$) and linearly increasing 
number of transmit antennas to achieve scalability. In contrast, all the scalable multicast-based schemes
listed in Theorem~\ref{theo:scalableSchemes} require only a fixed total power and a reduced number of
transmit antennas.
\end{remark}

\section{Further improvement with simultaneous multicasting and multiplexing}
\label{sec:mixed}
\newcommand{\Rmx}{\mathcal{R}_{\rm {mix}}}
\newcommand{\whp}{\emph{w.h.p.}}
\newcommand{\Rmix}{\bar{R}^{\rm {mix}}}
\newcommand{\Rsymmix}{\Rsym^{\rm {mix}}}

In previous sections, we have investigated two extreme uses of multiple antennas: (coded) multicasting and spatial multiplexing. The gains of these two techniques are pronounced in different regimes. Spatial multiplexing achieves good performance with precise CSIT and at high power. Whereas, at a fixed total power and without CSIT, multicasting can still achieve the scalability of the system. 
Therefore, it is favorable to perform simultaneous spatial multiplexing and multicasting to further benefit from both multiplexing gain and global caching gain. The synergy of multicasting and spatial multiplexing in coded caching was observed in \cite{zhang2015coded,Elia2017} for minimizing the transmission time at high SNR regime. It was shown that when CSIT is perfectly known, ZF with uncoded caching is optimal for that purpose, since ZF can eliminate inter-user interference and create parallel links using perfect channel knowledge. When the CSIT is imperfect, however, the interference is inevitable, and coded caching is needed to retrieve the minimal transmission time. In our setting, the simultaneous multicasting and multiplexing can be done with rate splitting\footnote{The combination of multicast and spatial multiplexing in the presence of CSIT error was first proposed in \cite{yang2013degrees} and then investigated in \cite{dai2015rate}~(and the references therein). This technique was first applied to coded caching in~\cite{zhang2015coded}.} as follows.


\subsection{Simultaneous multicasting and multiplexing}
We consider the transmission of signal carrying both the {\em common} information coded in $\rvVec{X}_0$
interested by all the users, and a set of \emph{private} information coded in $\{\rv{X}_k\}$ where $\rv{X}_k$ is intended exclusively for user $k$, $k \in [K]$. We still consider quasi-static fading channels. 
The transmitted signal is
\begin{align}
\rvVec{X} = \rvVec{X}_0 + \sum_{k=1}^K \rvVec{W}_k X_k,
\end{align}
where $\rv{X}_0, \rv{X}_k, \rvVec{W}_k$, $k \in [K]$ are defined as before, except for the new total power constraint $\sum_{k=0}^K P_k \leq P$. Obviously, this general setting includes the two extreme cases $P_0 = 0$ for spatial multiplexing and $P_0 = P$ for multicasting. 
We use the same assumption $\nt\ge K$ as in the previous section. 
The received signal at user~$k$ is
\begin{align}
\rv{Y}_k &= \rvVec{H}_k^\T \rvVec{X}_0 + \rv{G}_k \rv{X}_k + \sum_{l\ne k} \tilde{\rv{G}}_{k,l} \rv{X}_l + \rv{Z}_k,
\end{align}%
where $\rv{G}_{k}$ and $\tilde{\rv{G}}_{k,l}$ are defined as for \eqref{eq:ZFrxsignal}. Each receiver is
interested in decoding the common message and its own private message. For simplicity, we consider
successive decoding so that each user decodes the common message first and then the private message.
Therefore, the private signals are seen as interference while decoding the common message. The SINR of the common signal at receiver~$k$ is
\begin{align}
\SINR_k^{(0)}(\rvMat{H}) &\defeq
\frac{\frac{P_0}{n_t}\|\rvVec{H}_k\|^2}{1+ {|\rv{G}_k|}^2 P_k + \sum_{l\ne k} |\tilde{\rv{G}}_{k,l}|^2 P_l }, \label{eq:SINR0}
\end{align}
and the long-term average common (multicast) rate is 
\begin{align}
\Rmix_0 = \E\left[ \ln\left(1 + \min_{k\in[K]} \SINR_k^{(0)}  \right) \right].
\end{align}
Then, the private messages are decoded as before after removing the decoded common signal, with the
same $\SINRzf$ as defined in \eqref{eq:SINRzf}, except that the power is reduced from $P$ to $P-P_0$. Let us consider uniform private power allocation $P_k = \frac{P-P_0}{K}, ~\forall k\in [K]$, then average symmetric private rate $\Rsymmix$ is defined similarly to $\Rsym$ accordingly. While the asymptotic behavior of $\Rsymmix$ is easy to characterize by following the
same steps as in the spatial multiplexing case, the analysis of $\Rmix_0$ is not trivial due to the interference terms in \eqref{eq:SINR0}. 


\begin{proposition} \label{prop:mixRate}
	Let us consider uniform private power allocation $P_k = \frac{P-P_0}{K}, \forall\,k\in[K]$ and assume that
	$\liminf\limits_{K\to\infty}\frac{\nt}{K} > 1$. When $(\nt-K+1)(1-\sigma^2) \to \infty$, the common rate $\Rmix_0$ and symmetric private rate $\Rsymmix$ scale as~\eqref{eq:R0mix} and~\eqref{eq:Rsymmix}, respectively.
	\begin{figure*}[!t]
	\begin{align} 
	\Rmix_0 &\sim \E[\ln\left(1+\frac{P_0}{1+\frac{P-P_0}{K}\left[(\nt-K+1)(1-\sigma^2)+
		\max_{k\in[K]} \left\{\sum_{l\ne k} |\tilde{\rv{G}}_{k,l}|^2 \right\}
		\right]}\right)], \label{eq:R0mix} \\
	\Rsymmix &\sim \ln\left( 1 +  \frac{(\nt-K+1)(1-\sigma^2) }{\frac{K}{P-P_0} + (K-1) \sigma^2 } \right). 
	\label{eq:Rsymmix}
	\end{align}%
		\setlength{\arraycolsep}{1pt}
	\hrulefill \setlength{\arraycolsep}{0.0em}
	\vspace*{1pt}
	\end{figure*}
\end{proposition}
Since the proof of this proposition does not provide additional insight in the problem, it is deferred to Appendix~\ref{app:mix}.

\subsection{Delivery scheme and equivalent content delivery rate}
The delivery scheme exploiting simultaneous multicasting and multiplexing operates as follows. For \textit{placement phase}, we use the same centralized or decentralized placement as in \cite{maddah2013fundamental} and \cite{maddah2013decentralized}, respectively. We still assume that each user $k$ has a sequence of demands $d_k^{(1)}, d_k^{(2)}, \dots, k\in[K]$, which are revealed to the server whenever a user requests a file. Next, each demand sequence is split into two subsequences to be delivered simultaneously in the \textit{delivery phase}. Assume that at the beginning, the server receives two different demands from each user $k$: $d^{(1)}_k$ in subsequence 1 and $d^{(2)}_k$ in subsequence 2. To deliver $\{d^{(1)}_k\}$, the server forms a multicast codeword containing $T(m,K)F$~bits following the centralized/decentralized coded caching scheme, and encodes it in the common signal $\rvVec{X}_0$. Meanwhile, it encodes the uncached fraction of $d^{(2)}_k$ containing $(1-m)F$~bits in the private signal $\rv{X}_k$, $k\in[K]$.\footnote{The encoding is across time, but we omitted the time index for simplicity.} The server then transmits $\rvVec{X} = \rvVec{X}_0 + \sum_{k=1}^K \rvVec{W}_k X_k$. Each user $k$ can get $d^{(1)}_k$ using its cache content and the multicast codeword decoded from the common signal, following the centralized/decentralized coded caching scheme. Next, since user $k$ already cache $mF$~bits of the file $d^{(2)}_k$, it can also get $d^{(2)}_k$ from its cache content and the message decoded from the private signal. The succeeding demands in the two subsequences of each user are served similarly.

Thus, each user can receive two independent flows. In the first flow carried in the common message, to deliver a file to each user with coded caching, we need to send $T(m,K)F$ bits, which takes $T(m,K)F/\Rmix_0$ units of time, so the equivalent content delivery rate is $\frac{1}{T(m,K)}\Rmix_0$. In the second flow carried in the private message, to deliver another file to each user, we need to send $(1-m)F$~bits, which takes $(1-m)F/\Rsymmix$ units of time, so the equivalent content delivery rate is $\frac{1}{1-m}\Rsymmix$. Since the two flows are independent, requested files can be delivered in parallel between two flows and consecutively within each flow.
Thus, the aggregated content delivery rate with the proposed scheme is simply the sum of the rates achieved with two flows
\begin{align} \label{eq:Rmix}
\Rmx = \frac{K}{T(m,K)}\Rmix_0 + \frac{K}{1-m}\Rsymmix.
\end{align}
The asymptotic behavior of $\Rmx$ depends on that of $\Rmix_0$ and $\Rsymmix$, which was provided in Proposition~\ref{prop:mixRate}. A practically relevant question is to find out the optimal power split $(P_0, P-P_0)$ that maximizes the content delivery rate
$\Rmx$. This problem is not trivial to solve, even in the large $K$ regime, due to the expectation in \eqref{eq:R0mix}. Let us relax this problem in the following example to understand the behavior of the optimal power split. 
\begin{example} \label{example:Popt}
  Let us assume that $\lim\limits_{K\to\infty}\frac{\nt}{K} = \beta > 1$, $(\nt-K+1)(1-\sigma^2) \to \infty$, and, for simplicity, remove the maximization in \eqref{eq:R0mix}. In this
  case, we can write the content delivery rate as $\Rmx \sim G(P,P_0)$ with 
\begin{multline}
G(P,P_0) \defeq \frac{K}{T(m,K)}\ln\left(1+\frac{ P_0}{1+(P-P_0)I_c}\right) \\+ \frac{K}{1-m}\ln\left( 1 +  \frac{I_c-I_p}{(P-P_0)^{-1} +
I_p} \right),
\end{multline}
where $I_c = \frac{(\nt - K +1)(1-\sigma^2) + (K-1)\sigma^2}{K} = \Theta(1)$ and $I_p = \frac{(K-1)\sigma^2}{K} = \Theta(\sigma^2)$. It follows that the optimal power
split should satisfy
	\begin{align} \label{eq:optimalP0}
	\hspace{-.7cm}P - P_0 \sim \left(\frac{-\frac{1-m}{T(m,K)}(1+I_cP)+(I_c-I_p)(1+ P)}{\frac{1-m}{T(m,K)}I_p(1+I_cP)-I_c(I_c-I_p)}\right)^+.
	\end{align}
\end{example}
\begin{remark} \label{remark:optimalP0}
Some properties of the optimal power split can be observed from \eqref{eq:optimalP0}. First, the optimal private power fraction $\frac{P-P_0}{P}$ is decreasing with total power $P$. 
That is, when the total power is low, spending more power to multicast is beneficial, and on the other hand, when the total
power is high, spatial multiplexing should be favored. 
Second, $\frac{P-P_0}{P}$ is decreasing with $m$. That is, more power should be allocated to multicast when the users' cache memory grows. 
This is reasonable since the global caching gain, which comes with multicasting and not with spatial multiplexing, scales up with user cache size. 
\end{remark}

When $(\nt-K+1)(1-\sigma^2) = O(1)$, the CSIT error $\sigma^2 \to 1$. Under this extremely low quality of channel estimate, it is rather clear that multicasting should be even more favored w.r.t. spatial multiplexing than in the case $(\nt-K+1)(1-\sigma^2) \to \infty$. 

\subsection{Numerical results}
In the rest of the section, we show some numerical results to illustrate the equivalent content delivery rate of the mixed delivery and the optimal power split. We consider the system having as many antennas as users, i.e., $\nt = K$. Note that, although we assumed asymptotically more antennas than users in Proposition~\ref{prop:mixRate} and Example~\ref{example:Popt}, the behavior of optimal power split when $\nt = K$ follows the same line of these analytical analysis, as can be observed shortly. Moreover, we consider a fixed per-user power $\frac{P}{K}$, and fixed CSIT error $\sigma^2 = \left(\frac{P}{K}\right)^{-1}$. 

%
First, in Fig.~\ref{fig:Rmul_uni_mix}, we compare the content delivery rate of mixed transmission with optimal power
split, spatial multiplexing alone, and coded multicasting (coded caching with multicasting) alone. 
We observe that optimal mixed transmission is always better than either scheme
alone. For example, about $50\%$ gain is achieved by mixed transmission w.r.t.~either scheme when
$m \approx 6.5\%$ and $P/K = 20$~dB. When $m$ is very small, spatial multiplexing is better than coded multicasting. On the other hand, when $m$ is moderate or large, coded multicasting is better. Further, when $m$ is larger than a certain ratio of the library, coded multicasting becomes optimal.

\begin{figure}[!h]
	\centering
	\subfigure[$P/K = 10$~dB]{
		\includegraphics[width=.45\textwidth]{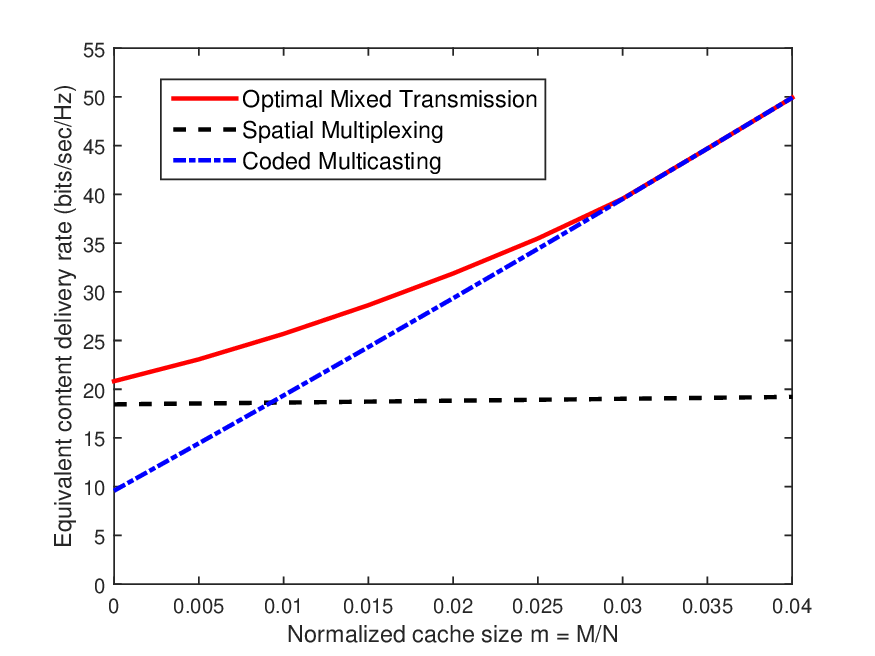}}
	\subfigure[$P/K = 20$~dB]{
		\includegraphics[width=.45\textwidth]{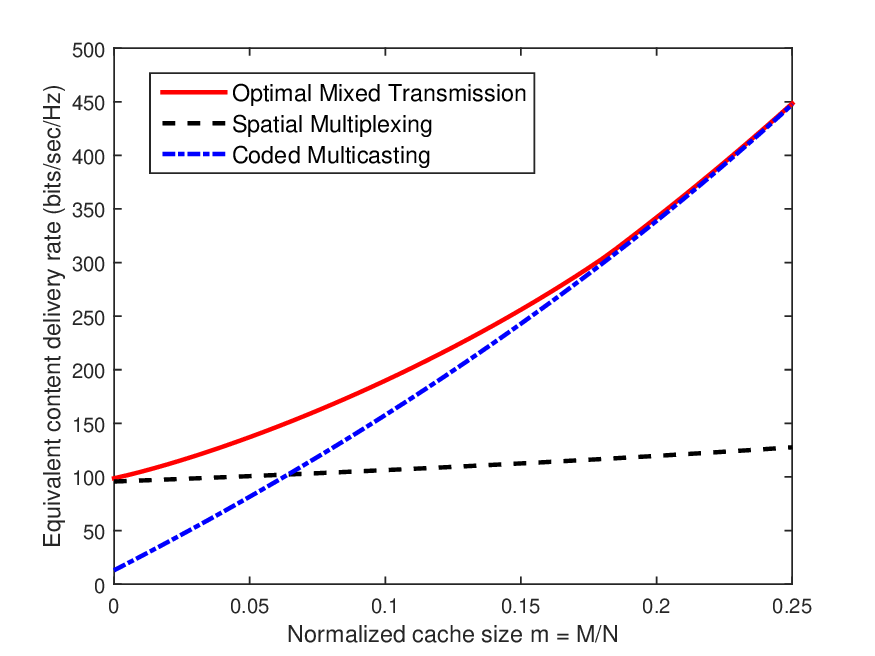}}
	\caption{The equivalent content delivery rate of optimal mixed transmission, spatial multiplexing and coded multicasting as a function of normalized cache size $m$ for $\nt = K = 100$, $P/K = 10,20$~dB, $\sigma^2 = \left(\frac{P}{K}\right)^{-1}$.}
	\label{fig:Rmul_uni_mix}
\end{figure}

Next, in Fig.~\ref{fig:Popt}, we plot the optimal common power fraction $P_0/P$ as a function of normalized cache size $m$ for different values of per-user power $P/K$. As $m$ increases, the figure suggests to allocate more power to multicasting, and even give all power to multicasting when $m$ is larger than a certain ratio of the library, namely, $3.5\%$ for $P/K = 10$~dB, $25\%$ for $P/K = 20$~dB, and $48\%$ for $P/K = 30$~dB. 
\begin{figure}[!h]
	\centering
	\includegraphics[width=.45\textwidth]{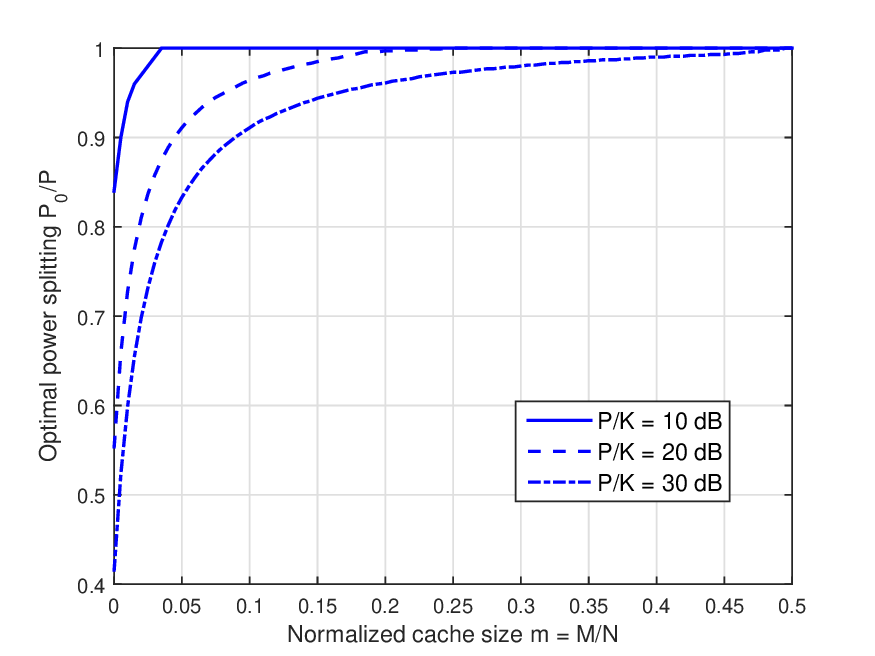}
	\caption{The optimal power splitting, represented by the common power fraction $P_0/P$, as a function of normalized cache size $m$ for $\nt =K = 100$, $P/K = 10, 20, 30$~dB, $\sigma^2 = \left(\frac{P}{K}\right)^{-1}$.}
	\label{fig:Popt}
\end{figure}
\begin{figure}[!h]
	\centering
	\includegraphics[width=.45\textwidth]{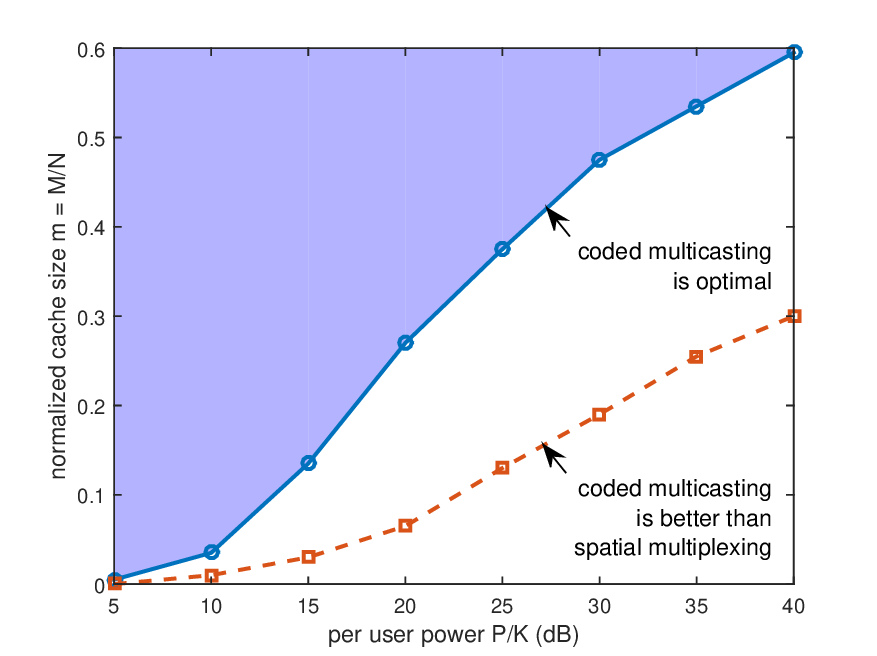}
	\caption{The preferable and optimal region (the set of pairs $(P/K,m)$) of coded multicasting for $\nt = K = 100$, $\sigma^2 = \left(\frac{P}{K}\right)^{-1}$: above the dashed line, coded multicasting is better than spatial multiplexing; above the solid line, coded multicasting is optimal.}
	\label{fig:m_opt}
\end{figure}

From the Fig.~\ref{fig:Rmul_uni_mix} and Fig.~\ref{fig:Popt}, we have observed that, for a given per-user power $P/K$, when the cache memory is sufficiently large, the optimal mixed transmission coincides with coded multicasting and there is no need for spatial multiplexing. This is further illustrated in Fig.~\ref{fig:m_opt}. For every pair $(P/K,m)$ in the shaded region (above the solid line) of the power-memory plane, coded multicasting is optimal, i.e., the optimal power split is $P_0/P = 1$. Besides, we also plot the values of normalized cache size $m$ over which coded multicasting outperforms spatial multiplexing and hence is preferable (the dashed line).


\section{Conclusion}
\label{sec:conclusion}

How to exploit multi-antenna downlink channels to achieve a scalable content delivery rate when the number of users goes to infinity? 
This is the main question that we have addressed in this work. Under various assumptions on the system configurations such as 
the number of transmit antennas, the number of coherence resource blocks, and the CSIT accuracy, we have investigated the multicast-based coded
caching schemes as well as the more conventional spatial multiplexing scheme. A general conclusion from the study is that multicast-based
coded caching is a more attractive option since linear rate scaling can be achieved without CSIT and with only sub-linear number of
transmit antennas with respect to the number of users. 
Based on rate splitting, we have also proposed to combine both multicast and spatial multiplexing to further improve the
performance. The effectiveness of such a combination has been confirmed with the numerical results. It is remarked that when the per-user power is small or the user cache memory is large enough, coded caching with multicasting is optimal and there is no need for spatial multiplexing. 

Due to the symmetry of the setting, we have obtained some simple analytical results in this work. Nevertheless, it would be interesting, in
the future, to consider the more general systems with different path loss, spatial correlation, and fairness constraints across users.

\appendix

\subsection{Proofs of the lemmas in Section~\ref{sec:preliminaries}}
\label{app:preliminaries}

\subsubsection{Proof of Lemma~\ref{lemma:CDFgamma}}
Since $\rv{X}\sim {\rm Gamma}\big({\nt,\frac{1}{\nt}}\big)$, $\rv{X}$ is equivalent to a sum of $\nt$ i.i.d.~exponential
random variables ${\rm Exp}(\nt)$. Thus, we can apply the Chernoff bound \eqref{eq:Chernoff} and obtain
\begin{align}
  \P[ \rv{X} \le x ] &\le e^{\nu x} \, \left(\E[e^{-\nu \rv{Z}}]\right)^{\nt} \\&= \frac{e^{\nu x}}{(1+\nu/\nt)^{\nt}}, \quad\forall\,\nu>0, \label{eq:tmp932}
\end{align}%
where $\rv{Z}\sim{\rm Exp}(\nt)$. It can be shown that for any $x<1$, $\nu^* = \nt(x^{-1}-1)>0$ minimizes the right hand
side of \eqref{eq:tmp932} which becomes $\P[ \rv{X} \le x ] \le e^{-\nt(x-1-\ln x)}$. Let $x = \const$ with
$\const\approx 0.1586$, we obtain \eqref{eq:CDFgamma}.

\subsubsection{Proof of Lemma~\ref{lemma:expectation}}
Let us define $\rv{Y}:= \min_{k\in[K]} \rv{X}_k$. It follows that the CDF of $\rv{Y}$ is $F_{\rv{Y}}(y) =
1-(1-F_{\rv{X}}(y))^K$. With Markov's inequality, we have for any $x_0>0$, $\E[\rv{Y}] \ge x_0 (1-F_{\rv{Y}}(x_0))$ from which
inequality \eqref{eq:Markov1} follows. If $F_{\rv{X}}(x)$ is strictly increasing, the inverse function
$F^{-1}_{\rv{X}}(x)$ exists. For any given $c>0$ and $K$ large enough, we have $\frac{c}{K}<1$ and let $x_0 =
F^{-1}_{\rv{X}}(\frac{c}{K})$. Then, applying \eqref{eq:Markov1}, we can prove \eqref{eq:Markov2} since 
  \begin{align}
    \frac{\E[ \min_{k\in[K]}\rv{X}_k]}{x_0} \ge
    \Bigl(1-\frac{c}{K}\Bigr)^{K} = e^{-c} + o(1), 
  \end{align}
when $K\to\infty$.

\subsubsection{Proof of Lemma~\ref{lemma:log-asympt}}
\underline{Case 1:} if $\E[\rv{X}_K] = \Theta(1)$, then there exists some $c>0$ and $1\ge \rho>0$ such that $\P[\rv{X}_K \ge c] \ge \rho$ when
$K\to\infty$. Otherwise, we would have $\E[\rv{X}_K] = o(1)$. Thus, with probability of at least $\rho$, we have
$\ln(1+\rv{X}_K) \ge \ln(1+c)$, from which $\E[\ln(1+\rv{X}_K)] \ge \rho \ln(1+c) = \Theta(1)$. This and the obvious upper bound $\E[\ln(1+\rv{X}_K)] \le \ln(1+\E[\rv{X}_K]) = \Theta(1)$ confirm $\E[\ln(1+\rv{X}_K)] = \Theta(1)$. 

\underline{Case 2:} if $\E[\rv{X}_K] = o(1)$ and $\E[\rv{X}_K^2] = o(\E[\rv{X}_K])$, then using $\ln(1+x) \ge x - \frac{x^2}{2}$ we can easily show that
$\E[\ln(1+\rv{X}_K)] \ge \E[\rv{X}_K] + o(\E[\rv{X}_K])$. Using Jensen's inequality, we also have $\E[\ln(1+\rv{X}_K)] \le
\ln(1+\E[\rv{X}_K]) \le \E[\rv{X}_K]$. Then $\E[\ln(1+\rv{X}_K)] \sim \E[\rv{X}_K]$.

\underline{Case 3:} $\E[\rv{X}_K] = \Theta(g(K))$ with $g(K) \to \infty$ and $\Pr\left(\rv{X}_K \ge g(K)\right) \ge \rho$ for some $\rho > 0$. From these conditions and Markov's inequality, it readily follows that $\E[\ln(1+\rv{X}_K)] \ge \ln(1+g(K)) \Pr\left(\rv{X}_K \ge g(K)\right) \ge \rho \ln(1+g(K))$. From Jensen's inequality, we also have 
$\E[\ln(1+\rv{X}_K)] \le \ln(1+\E[\rv{X}_K]) = \Theta(\ln(1+g(K)))$, which completes the proof.

\subsubsection{Proof of Lemma~\ref{lemma:SNR0}}
\label{proof:lemmaSNR0}

To prove the convergence of mean, it is enough to show the {\em convergence in distribution} and the {\em uniform integrability}
\cite[Theorem~3.5]{billingsley1999convergence}. Let us define $a_K \defeq \nt\left(\frac{K}{\nt!}\right)^{1/\nt}$. 

First, we shall
show that the sequences $\left\{a_K\min_{k\in[K]} \normh\right\}_K$ and $\left\{ \left(a_K\min_{k\in[K]} \normh \right)^2
\right\}_K$ converge in distribution to the random variables $\rv{Y}$ and $\rv{Y}^2$, respectively, where the random variable
$\rv{Y}$ has CDF $F_\rv{Y}(y) = 1-e^{-y^\nt}$. 
To that end, we focus on the convergence of $a_K\min_{k\in[K]} \normh$, from which the convergence of $\left(a_K\min_{k\in[K]}\normh \right)^2$
can be shown with the {\em continuous mapping theorem}. The proof follows essentially the footsteps of~\cite[Theorem 1]{Hassibi2008scalinglaw}
  and is provided here to be self-contained. Denote $\rv{X}_k \defeq \normh$, then $\rv{X}_k$ is i.i.d. ${\rm
  Gamma}\left(\nt,\frac{1}{\nt}\right)$ across $k$ with CDF $F_\rv{X}(x) = \frac{\gamma(\nt,\nt x)}{\Gamma(\nt)} = 1 - e^{-\nt x}
  \sum_{i = 0}^{\nt-1} \frac{(\nt x)^i}{i!}$ and $\rv{X}_{\min} \defeq \min_{k\in[K]} \rv{X}_k$ has the CDF $F_{\min}(x) =
  1-(1-F_\rv{X}(x))^K$. Expanding $F_\rv{X}(x)$ in Taylor series yields $F_\rv{X}(x) = \sum_{i = 0}^\infty F_\rv{X}^{(i)}(0)
  \frac{x^i}{i!}$, where $F_\rv{X}^{(i)}(0) = (-1)^{i-\nt}\nt^i \binom{i-1}{\nt-1}$ if $i\geq\nt$ and $0$ otherwise. 
	Then 
	\begin{align}
	F_{\min}(x) = 1 - \left[1-\frac{F_\rv{X}^{(\nt)}(0)}{\nt!}x^\nt - \underset{i>\nt}{\sum}\frac{F_\rv{X}^{(i)}(0)}{i!}x^i\right]^K.
	\end{align} 
	Replacing $x$ by $\left(\frac{\nt!}{KF_\rv{X}^{(\nt)}(0)}\right)^{\frac{1}{\nt}} x = \frac{1}{\nt}\left(\frac{\nt!}{K}\right)^{\frac{1}{\nt}} x = \frac{x}{a_K}$, we obtain
	\begin{align}
	&F_{\min}\left(\frac{x}{a_K}\right) \notag \\&= 1-\left[1-\frac{x^\nt}{K}-\sum_{i>\nt}(-1)^{i-\nt}\binom{i-1}{\nt-1}\frac{1}{i!}\left(\frac{\nt!}{K}\right)^{i/nt}x^{i}\right]^K \\
	&= 1-e^{-x^\nt} + o(1), 
	\end{align}
	since $\sum_{i>\nt}(-1)^{i-\nt}\binom{i-1}{\nt-1}\frac{1}{i!}\left(\frac{\nt!}{K}\right)^{i/nt}x^{i} =
        \Theta\left(K^{-1-\frac{1}{\nt}}\right)$ vanishes faster than $\frac{x^\nt}{K} = \Theta(K^{-1})$ and
        $\left(1-\frac{x^\nt}{K}\right)^K = e^{-x^\nt} + o(1)$. From this, a simple change of variable $\rv{Y}_K =
	    a_K\rv{X}_{\min}$ gives $F_\rv{Y_K}(y) \xrightarrow{K\to\infty} 1-e^{-y^\nt}$. 

        Then, we shall show both sequences
  $\left\{a_K{\displaystyle\min_{k\in[K]}}\normh \right\}_K$ and
  $\left\{\left(a_K{\displaystyle\min_{k\in[K]}} \normh \right)^2\right\}_K$ are uniformly integrable.
  \newcommand{\indU}[1]{{\mathbbm{1}_{\{\rv{U}_{#1}\ge \omega\}}}}
	Let $\rv{U}_K \defeq a_K\min_{k\in[K]} \normh$. It is enough to show that $\{\rv{U}_K\}_K$ satisfies $\lim\limits_{\omega \to \infty} \sup_K \E \left[\rv{U}_K \indU{K} \right] = 0$. 
        Indeed,	
	\begin{align}
	\E \left[\rv{U}_K \indU{K} \right] &= \int_{\omega}^{\infty} u dF_{\rv{U}_K}(u) \\
        &= \omega [1-F_{\rv{U}_K}(\omega)] + \int_{\omega}^{\infty} [1-F_{\rv{U}_K}(u)]du.
        \label{eq:tmp37} 
	\end{align}
        which cannot increase with $K$ since $1-F_{\rv{U}_K}(x) = 1-F_{\min_{k\in[K]} \normh}(\frac{x}{a_K}) =
        \left[\frac{\Gamma(\nt,\frac{\nt x}{a_K})}{\Gamma(\nt)}\right]^K$ is non-increasing with $K$. Therefore,
	\begin{align}
	\sup_K \E[\rv{U}_K \indU{K}] &= \E[\rv{U}_1 \indU{1}] \\&= a_1 \frac{\Gamma\left(\nt+1,\nt\omega/a_1\right)}{\Gamma(\nt+1)} 
	\xrightarrow{\omega\to\infty} 0,
	\end{align}
	which means that $\{\rv{U}_K\}$ is uniformly integrable. 
	\newcommand{\indUs}[1]{{\mathbbm{1}_{\{\rv{U}^2_{#1}\ge \omega\}}}}
        Similarly, 
	$\lim\limits_{\omega \to \infty} \sup_K \E[\rv{U}^2_K \indU{K}] = 0$, 
	 since 
	\begin{align}
	\sup_K \E[\rv{U}^2_K \indU{K}] &= \E[\rv{U}^2_1 \indU{1}] \\&= a_1^2 \frac{\Gamma\left(\nt+2,\nt\omega/a_1\right)}{\nt\Gamma(\nt+1)} 
	\xrightarrow{\omega\to\infty} 0.
	\end{align} 
	Thus $\{\rv{U}_K^2\}$ is also uniformly integrable. 
	
	Explicit calculation of $\E[\rv{Y}]$ and $\E[\rv{Y}^2]$ completes the proof of Lemma~\ref{lemma:SNR0}.

\subsubsection{Proof of Lemma~\ref{lemma:SNR01}}

We begin by proving the first part of the lemma. 
Since both $\E \left[\displaystyle\min_{k\in[K]} \normh\right]$ and  
$\E \left[\left(\displaystyle\min_{k\in[K]} \normh\right)^2\right]$ are upper bounded, which can be seen by removing the ``min'' inside the
expectation, it is enough the show that they are also lower bounded. Let us define $\rv{X}_k \defeq \frac{\| \rvVec{H}_k
\|^2}{\nt}\sim{\rm Gamma}(\nt,\frac{1}{\nt})$, $k\in[K]$, with common CDF $F_{\rv{X}}(x)$. From Lemma~\ref{lemma:CDFgamma}, we have
$F_{\rv{X}}(\const) \le e^{-\nt}$ with $\const \simeq 0.1586$. When $\ntgelnK$, we have $\nt \ge \ln(K)+c_0$ for some
$c_0>-\infty$ when $K$ is large enough. It follows that $F_{\rv{X}}(\const) \le e^{-\ln(K)-c_0} = \frac{e^{-c_0}}{K}$, and
consequently $F^{-1}_{\rv{X}}(\frac{e^{-c_0}}{K}) \ge \const$ due to the monotonicity of $F_{\rv{X}}(x)$. Then,
\begin{align}
  \hspace{-.5cm}\frac{\E[ \min_{k\in[K]}\rv{X}_k]}{\const} \ge \frac{\E[ \min_{k\in[K]}\rv{X}_k]}{F^{-1}_{\rv{X}}(\frac{c}{K})} \ge
  e^{-c} + o(1), \label{eq:tmp787}
\end{align}
when $K\to\infty$, where the second inequality is from \eqref{eq:Markov2} for $c:=e^{-c_0}$. This completes the proof of \eqref{eq:tmp888}. To prove
\eqref{eq:tmp889}, it suffices to apply $\E[\rv{X}^2]\ge \E[\rv{X}]^2$ and \eqref{eq:tmp787}, and we have
\begin{align}
  \frac{\E[ \left( \min_{k\in[K]}\rv{X}_k \right)^2]}{\const^2} \ge \frac{\E[
  \min_{k\in[K]}\rv{X}_k]^2}{\left(F^{-1}_{\rv{X}}(\frac{c}{K})\right)^2} \ge
  (e^{-c} + o(1))^2, 
\end{align}
when $K\to\infty$.

For the second part, it suffices to show that, when $\ln(K) = o(\nt)$, for any given $\epsilon>0$, both $\P[\min_{k\in[K]} \normh \ge
1+\epsilon]$ and $\P[\min_{k\in[K]} \normh \le 1-\epsilon]$ go to $0$ when $K\to\infty$. To that end, we first bound $\P[\min_{k\in[K]}
\normh \ge 1+\epsilon] \le \P[ \normh \ge 1+\epsilon]$ which is then upper bounded by $e^{-\nu (1+\epsilon)} (1-\frac{\nu}{\nt})^{-\nt}$
for any $0< \nu<\nt$ from the Chernoff bound \eqref{eq:Chernoff2}. Letting $v = \nt(1-(1+\epsilon)^{-1})$, the upper bound becomes
$e^{-\nt(\epsilon-\ln(1+\epsilon))}$ which goes to $0$ for any $\epsilon>0$. Now, let us consider $\P[\min_{k\in[K]} \normh \le
1-\epsilon]$ which can be rewritten as $1-\left(1-\P[\normh \le 1-\epsilon]\right)^K$.  From the Chernoff bound \eqref{eq:Chernoff}, we
have $\P[\normh \le 1-\epsilon] \le e^{\nu (1-\epsilon)} (1+\frac{\nu}{\nt})^{-\nt}$ for any $\nu > 0$. Letting $v = \nt((1-\epsilon)^{-1} - 1)$ which minimizes the upper bound, we obtain
$\P[\normh \le 1-\epsilon] \le e^{-\nt (-\ln(1-\epsilon) - \epsilon)} = e^{-\nt \delta_\epsilon}$ with $\delta_\epsilon :=
-\ln(1-\epsilon) - \epsilon>0$ for $\epsilon>0$. Therefore, we have 
\begin{align}
 \hspace{-.7cm}\P[\min_{k\in[K]} \normh \le 1-\epsilon] &\le 1 - (1-e^{-\nt\delta_\epsilon})^K \\
  &= 1 - e^{K \ln(1-e^{-\nt\delta_\epsilon})} \\
  &= 1 - e^{- K e^{-\nt\delta_\epsilon} + K o(e^{-\nt\delta_\epsilon})}. \label{eq:tmp221}
\end{align}%
Since $\ln(K) = o(\nt)$, $K e^{-\nt\delta_\epsilon} = e^{\ln(K) - \nt \delta_\epsilon} \to 0$ for any $\delta_\epsilon$. We have just proved that the upper
bound \eqref{eq:tmp221} goes to $0$, which completes the proof. 

\renewcommand*{\proofname}[1]{Proof}

\subsection{Proof of Lemma~\ref{lemma:SINRzf}}
	\label{proof:lemmaSINRzf}
We recall the definition of the precoding vector for user $k$, $	\rvVec{W}_k = \alpha_k \rvMat{U}_k \rvMat{U}_k^\H \hat{\rvVec{H}}_k^*$,
where the columns of $\rvMat{U}_k$ form an orthonormal basis of the null space of ${\rm span}(\{\hat{\rvVec{H}}^*_l\}_{l\ne k})$ and $\rvMat{U}_k$ is independent of $\hat{\rvVec{H}}_k$; $\alpha_k \defeq 1/\|\hat{\rvVec{H}}_k^\T \rvMat{U}_k\|$. With uniform power allocation, $\SINRzf(\rvMat{H}) \defeq \frac{{|\rv{G}_k|}^2}{p^{-1}+ \sum_{l\ne k} |\tilde{\rv{G}}_{k,l}|^2}$,
where $\rv{G}_{k} \defeq {\rvVec{H}}_k^\T \rvVec{W}_k = 
\hat{\rvVec{H}}_k^\T \rvVec{W}_k + \tilde{\rvVec{H}}_k^\T \rvVec{W}_k$ and
$\tilde{\rv{G}}_{k,l} \defeq \tilde{\rvVec{H}}_k^\T \rvVec{W}_l \sim \mathcal{CN}(0,	\sigma^2).$

First, consider $G_k$. Note that $\hat{\rvVec{H}}_k^\T \rvVec{W}_k = \|\hat{\rvVec{H}}_k^\T \rvMat{U}_k\|$ and $
\tilde{\rvVec{H}}_k^\T \rvVec{W}_k = \Bigl(\tilde{\rvVec{H}}_k^\T \rvMat{U}_k\Bigr) \Bigl(\hat{\rvVec{H}}_k^\T \rvMat{U}_k / \|\hat{\rvVec{H}}_k^\T \rvMat{U}_k\| \Bigr)^\H $. Since $\hat{\rvVec{H}}_k^\T$ and $\tilde{\rvVec{H}}_k^\T$ are independent and both contain i.i.d.~circularly
symmetric Gaussian variables, $\hat{\rvVec{H}}_k^\T \rvMat{U}_k$ and $\tilde{\rvVec{H}}_k^\T\rvMat{U}_k$ are also independent and have the same property. Further, the norm $\|\hat{\rvVec{H}}_k^\T \rvMat{U}_k\|$ and
the direction $\hat{\rvVec{H}}_k^\T \rvMat{U}_k / \|\hat{\rvVec{H}}_k^\T \rvMat{U}_k\|$ are independent for a vector of i.i.d.~circularly symmetric Gaussian variables. It readily follows that $\hat{\rvVec{H}}_k^\T \rvVec{W}_k$ and $\tilde{\rvVec{H}}_k^\T \rvVec{W}_k$ are indeed independent with 
\begin{align}
	\rv{A}_K &\defeq \tilde{\rvVec{H}}_k^\T \rvVec{W}_k \sim \mathcal{CN}(0,\sigma^2), \\
	\rv{B}_K &\defeq \frac{|\hat{\rvVec{H}}_k^\T \rvVec{W}_k|^2}{(\nt-K+1)(1-\sigma^2)} \\&\sim {\rm Gamma}\Bigl(\nt-K+1,\frac{1}{\nt-K+1}\Bigr), \label{eq:tmp166}
\end{align}%
where \eqref{eq:tmp166} is because $|\hat{\rvVec{H}}_k^\T \rvVec{W}_k|^2 \sim \Cc\Nc(0,(1-\sigma^2)\Id_{\nt-K+1})$ since $|\hat{\rvVec{H}}_k^\T \rvVec{W}_k|^2 = \|\hat{\rvVec{H}}_k^\T \rvMat{U}_k\|^2$ and $\rvMat{U}_k \in \CC^{\nt \times (\nt - K +1)}$ is independent of $\hat{\rvVec{H}}_k \sim \Cc\Nc(0,(1-\sigma^2) \Id_\nt)$.
If $\liminf\limits_{K\to\infty}\frac{\nt}{K} > 1$, we have 
$\rv{B}_K \xrightarrow{\text{a.s.}} 1$ by the strong law of large number.

Next, we consider the sum $\displaystyle\sum_{l\ne k} |\tilde{\rv{G}}_{k,l}|^2 =
\tilde{\rvVec{H}}_k^\T \Bigl(\sum_{l\ne k} \rvVec{W}_l\rvVec{W}_l^\H \Bigr) \tilde{\rvVec{H}}_k^*$. Let $\rv{C}_K \defeq \frac{1}{(K-1)\sigma^2}\displaystyle\sum_{l\ne k} |\tilde{\rv{G}}_{k,l}|^2$, then $\E[\rv{C}_K] = 1$. The
matrix $\rvMat{Q}_k:=\sum_{l\ne k} \rvVec{W}_l\rvVec{W}_l^\H $ is independent of
$\tilde{\rvVec{H}}_k^\T$ and has at most $K-1$ non-zero eigenvalues. Let $\rvMat{Q}_k = \rvMat{V}
\rvMat{\Lambda} \rvMat{V}^\H$ be the eigenvalue decomposition with $\rvMat{V}\in\mathbb{C}^{\nt\times (K-1)}$
being orthogonal and $\trace(\rvMat{\Lambda}) = K-1$. Then, $\sum_{l\ne k} |\tilde{\rv{G}}_{k,l}|^2 =
\breve{\rvVec{H}}_k^{\T} \rvMat{\Lambda} \breve{\rvVec{H}}_k^*$ where $\breve{\rvVec{H}}_k :=
\rvMat{V} \tilde{\rvVec{H}}_k$ contains $K-1$ i.i.d.~$\mathcal{CN}(0,\sigma^2)$ entries and is
independent of $\rvMat{\Lambda}$. With the assumption that $\liminf\limits_{K\to\infty}\frac{\nt}{K} > 1$, we
can show that the eigenvalues of $\rvMat{\Lambda}$ is bounded almost surely. To that end, let us write the precoding
matrix in an alternative form, namely, 
\begin{align}
  \rvMat{W} &:= [\rvVec{W}_1 \ \cdots \ \rvVec{W}_K] = \hat{\rvMat{H}}^\dagger \rvMat{D}
\end{align}%
where $\hat{\rvMat{H}}^\dagger := \hat{\rvMat{H}}^\H (\hat{\rvMat{H}}\hat{\rvMat{H}}^\H)^{-1}$ is the pseudo-inverse of the channel matrix
$\hat{\rvMat{H}}$ whereas $\rvMat{D}$ is a diagonal matrix with the $k$-th diagonal element, $\rv{D}_k$, normalizes the norm of
the $k$-th column of $\rvMat{H}^\dagger$. Since the norm of each column of $\hat{\rvMat{H}}^\dagger$ is lower bounded by the minimum
eigenvalue $\lambda_{\min}( (\hat{\rvMat{H}}^\dagger)^\H \hat{\rvMat{H}}^\dagger) = \lambda_{\min}((\hat{\rvMat{H}}\hat{\rvMat{H}}^\H)^{-1}) =
\lambda_{\max}(\hat{\rvMat{H}}\hat{\rvMat{H}}^\H)^{-1}$, we have 
\begin{align}
  \rv{D}_k &\le \lambda_{\max}(\hat{\rvMat{H}}\hat{\rvMat{H}}^\H), \quad \forall\,k \in [K]. 
\end{align}%
Consequently, we have 
\begin{align}
  \lambda_{\max}(\rvMat{W}^\H\rvMat{W}) &\le 
  \lambda_{\max}( (\hat{\rvMat{H}}^\dagger)^\H \hat{\rvMat{H}}^\dagger) \lambda_{\max}( \rvMat{D}^\H \rvMat{D} ) \le
  \frac{\lambda_{\max}(\hat{\rvMat{H}}\hat{\rvMat{H}}^\H)}{\lambda_{\min}(\hat{\rvMat{H}}\hat{\rvMat{H}}^\H)}
\end{align}%
which is upper bounded almost surely when $\liminf\limits_{K\to\infty}\frac{\nt}{K} > 1$ according to \cite{bai1998}. 
Following the footsteps in~\cite[Lemma 4]{WagnerSlock}, we can show that
\begin{align}
  \frac{1}{(K-1)\sigma^2}{\breve{\rvVec{H}}_k^\T} \rvMat{\Lambda}
  {\breve{\rvVec{H}}_k^*} - \frac{1}{K-1} \trace (\rvMat{\Lambda}) \xrightarrow{\text{a.s.}} 0,
\end{align}
which reads 
  $\rv{C}_K = \frac{1}{(K-1)\sigma^2}\sum_{l\ne k} |\tilde{\rv{G}}_{k,l}|^2 \xrightarrow{\text{a.s.}} 1$.

\subsection{Proof of Proposition~\ref{prop:mixRate}}
\label{app:mix}

\newcommand{\sK}{S_K}
\newcommand{\fK}{\rv{F}_{K,k}}

Since \eqref{eq:Rsymmix} follows readily from Proposition~\ref{prop:Rsym} by replacing $p$ by $\frac{P-P_0}{K}$, we focus on 
\eqref{eq:R0mix}. Due to the space limitation, we omit some of the technical details and only provide a sketch of proof. 

First, we have \eqref{eq:tmp575} and \eqref{eq:tmp576}.
\begin{figure*}[!t]
	\begin{align}
	\min_{k\in[K]} \SINR_k^{(0)} &\le \frac{P_0 \, {\displaystyle\max_{k\in[K]}} \{\frac{1}{n_t}\|\rvVec{H}_k\|^2\} }{1+ \frac{P-P_0}{K}
		\left( (\nt-K+1)\, {\displaystyle\min_{k\in[K]}} \{\frac{1}{\nt-K+1}{|\rv{G}_k|}^2\}  + {\displaystyle\max_{k\in[K]}} \{\sum_{l\ne k}
		|\tilde{\rv{G}}_{k,l}|^2 \}\right) }, \label{eq:tmp575} \\
	\min_{k\in[K]} \SINR_k^{(0)} &\ge \frac{P_0 \, {\displaystyle\min_{k\in[K]}} \{\frac{1}{n_t}\|\rvVec{H}_k\|^2\} }{1+ \frac{P-P_0}{K} \left( (\nt-K+1)\, {\displaystyle\max_{k\in[K]}} \{\frac{1}{\nt-K+1}{|\rv{G}_k|}^2\}  + {\displaystyle\max_{k\in[K]}} \{\sum_{l\ne k} |\tilde{\rv{G}}_{k,l}|^2 \}\right) }. \label{eq:tmp576}
	\end{align}%
	\setlength{\arraycolsep}{1pt}
	\hrulefill \setlength{\arraycolsep}{0.0em}
	\vspace*{1pt}
\end{figure*}
Then, from \eqref{eq:fast} in Lemma~\ref{lemma:SNR01}, we have $\min_{k\in[K]} \normh \xrightarrow{\text{p}} 1.$ In fact, following the
same proof of \eqref{eq:fast}, one can show that $\max_{k\in[K]} \normh \xrightarrow{\text{p}} 1$ since $\nt = \Omega(K)$. For the same reason,
we can show that ${\displaystyle\max_{k\in[K]}} \{\frac{1}{\nt-K+1}{|\rv{G}_k|}^2\}\xrightarrow{\text{p}} 1-\sigma^2$ and ${\displaystyle\min_{k\in[K]}} \{\frac{1}{\nt-K+1}{|\rv{G}_k|}^2\}\xrightarrow{\text{p}} 1-\sigma^2$ since $|\rv{G}_k|^2 =
|\rvVec{H}_k^\T \rv{W}_k|^2 \overset{\text{d}}{=} |\sigma \rv{A}_K + \sqrt{(\nt-K+1)(1-\sigma^2) \rv{B}_K}|^2$ with $\rv{A}_K$ and
$\rv{B}_K$ defined as in Lemma~\ref{lemma:SINRzf}. Indeed, we need to apply the assumption $(\nt-K+1)(1-\sigma^2)\to\infty$ to get rid of
the impact of $\rv{A}_K$ and the assumption $\nt-K+1 = \Omega(K)$ to obtain the convergence in probability. Therefore, both
the upper and lower bounds \eqref{eq:tmp575} and \eqref{eq:tmp576} tend to the same random variable in probability. This leads to the convergence in probability of $\ln (1+\min_{k\in[K]} \SINR_k^{(0)})$ by continuous mapping theorem.

Finally, we can prove that $\ln (1+\min_{k\in[K]} \SINR_k^{(0)})$ is uniformly integrable to get the convergence of mean, as is done
in Appendix~\ref{proof:lemmaSNR0}.

\bibliographystyle{IEEEtran}
\bibliography{IEEEabrv,./biblio}
\end{document}